\documentclass[a4paper,USenglish,autoref]{lipics-v2021}
\usepackage{microtype}
\usepackage{amsmath,amssymb}
\usepackage{xspace}
\usepackage{color}
\usepackage{graphicx}
\usepackage{algorithm}
\usepackage[noend]{algpseudocode}
\nolinenumbers 
\hideLIPIcs
\usepackage{hyperref}

\newcommand{\BeginMyItemize}{\begin{itemize}}%{\begin{itemize}\setlength{\itemsep}{-\parskip}}
\newcommand{\EndMyItemize}{\end{itemize}}

\newcommand{\BeginMyEnumerate}{\begin{enumerate}\setlength{\itemsep}{-\parskip}}
\newcommand{\EndMyEnumerate}{\end{enumerate}}

\renewcommand{\leq}{\leqslant}
\renewcommand{\geq}{\geqslant}

\newcommand{\Reals}{\mathbb{R}}

\newcommand{\Integers}{\mathbb{Z}}

\newcommand{\bd}{\partial\hspace*{0.5mm}}
\renewcommand{\preceq}{\preccurlyeq}

\newcommand{\Exp}{\mathrm{Exp}}

\newcommand{\A}{\ensuremath{\mathcal{A}}}
\newcommand{\B}{\ensuremath{\mathcal{B}}}

\newcommand{\M}{\ensuremath{\mathcal{M}}}
\newcommand{\R}{\ensuremath{\mathcal{R}}}
\newcommand{\sepset}{\ensuremath{\mathcal{S}}}
\newcommand{\PP}{\ensuremath{\mathbb{P}}}
\newcommand{\EE}{\ensuremath{\mathbb{E}}}

\newcommand{\ceil}[1]{\left\lceil #1 \right\rceil}
\newcommand{\ceili}[1]{\lceil #1 \rceil}
\newcommand{\br}[1]{\left( #1 \right)}
\newcommand{\floor}[1]{\left\lfloor #1 \right\rfloor}
\newcommand{\floori}[1]{\lfloor #1 \rfloor}

\newcommand{\eps}{\varepsilon}
\newcommand{\mydef}{:=}
\newcommand{\etal}{\emph{et al.}}

\newcommand{\length}{\mathrm{length}}

\newcommand{\tspp}{\mbox{{\sc Traveling Salesman Problem}}\xspace}

\newcommand{\etsp}{\mbox{{\sc Euclidean TSP}}\xspace}

\newcommand{\btsp}{\mbox{{\sc Bitonic TSP}}\xspace}
\newcommand{\bdtsp}{\mbox{{\sc Euclidean Path Cover}}\xspace}
\newcommand{\NP}{\mbox{{\sc np}}\xspace}

\newcommand{\fail}{{\sc Fail}\xspace}
\newcommand{\success}{{\sc Success}\xspace}

\newcommand{\yrange}{\mbox{$y$-range}}
\DeclareMathOperator{\ton}{ton}
\newcommand{\myright}{\mathrm{right}}
\newcommand{\myleft}{\mathrm{left}}
\newcommand{\ball}{\mathrm{Ball}}

\newcommand{\cC}{\mathcal{C}}

\newcommand{\cB}{\B}
\newcommand{\cM}{\M}
\newcommand{\cR}{\R}
\newcommand{\cP}{\mathcal{P}}
\newcommand{\sig}{\sigma}
\newcommand{\symdiff}{\triangle}
\newcommand{\poly}{\mathrm{poly}}
\newcommand{\myopt}{\mathrm{opt}}
\newcommand{\weight}{\mathrm{weight}}

\newcommand{\ovf}{\overline{F}}
\newcommand{\labda}{\lambda}

\newcommand{\sps}{s}
\newcommand{\SPS}{\mathcal{S}}
\newcommand{\spt}{t}
\newcommand{\SPT}{\mathcal{T}}
\newcommand{\spu}{t^+}
\newcommand{\SPU}{\mathcal{T}^+}
\newcommand{\ind}{\mathrm{st}}
\newcommand{\indp}{st^+\mkern-2mu}

\newcommand{\mypara}[1]{\medskip\noindent{\sf\textbf{#1}}}

\newcommand{\opt}{T_{\mathrm{opt}}}

\newcommand{\opteen}{T_{\mathrm{opt1}}}
\newcommand{\opttwee}{T_{\mathrm{opt2}}}

\title{Euclidean TSP in Narrow Strips}
\funding{The work in this paper is supported by the Netherlands Organisation for Scientific Research (NWO) through Gravitation-grant NETWORKS-024.002.003.\\
We would like to thank the anonymous DCG reviewer for their exceptionally thorough review and insightful comments.}
\author{Henk Alkema}{Department of Mathematics and Computer Science, TU Eindhoven, the Netherlands}{h.y.alkema@tue.nl}{}{}
\author{Mark de Berg}{Department of Mathematics and Computer Science, TU Eindhoven, the Netherlands}{m.t.d.berg@tue.nl}{}{}
\author{Remco van der Hofstad}{Department of Mathematics and Computer Science, TU Eindhoven, the Netherlands}{r.w.v.d.hofstad@tue.nl}{}{}
\author{S\'andor Kisfaludi-Bak}{Aalto University, Espoo, Finland}{sandor.kisfaludi-bak@aalto.fi}{}{}

\authorrunning{H.~Alkema, M.~de Berg, R.~van der Hofstad, and S.~Kisfaludi-Bak}
\Copyright{Henk Alkema, Mark de Berg, Remco van der Hofstad, S\'andor Kisfaludi-Bak}%mandatory, please use full first names. LIPIcs license is "CC-BY";  http://creativecommons.org/licenses/by/3.0/


% mandatory: Please provide 1-5 keywords
\EventEditors{John Q. Open and Joan R. Access}
\EventNoEds{2}
\EventLongTitle{42nd Conference on Very Important Topics (CVIT 2016)}
\EventShortTitle{CVIT 2016}
\EventAcronym{CVIT}
\EventYear{2016}
\EventDate{December 24--27, 2016}
\EventLocation{Little Whinging, United Kingdom}
\EventLogo{}
\SeriesVolume{42}
\ArticleNo{23}
\begin{document}
\maketitle
%------------------------------------------------------------------------------------------

\begin{abstract}
We investigate how the complexity of \etsp for point sets $P$ inside the strip $(-\infty,+\infty)\times [0,\delta]$ depends on the strip width~$\delta$.
We obtain two main results.
\begin{itemize}
\item For the case where the points have distinct integer $x$-coordinates, we prove that a shortest bitonic tour (which can be computed in $O(n\log^2 n)$ time using an existing algorithm) is guaranteed to be a shortest tour overall when $\delta\leq 2\sqrt{2}$, a bound which is best possible.
\item We present an algorithm that is fixed-parameter tractable with respect to~$\delta$.
Our algorithm has running time $2^{O(\sqrt{\delta})} n + O(\delta^2 n^2)$ for sparse point sets, where each $1\times\delta$ rectangle inside the strip contains $O(1)$ points.
For random point sets, where the points are chosen uniformly at random from the rectangle~$[0,n]\times [0,\delta]$, it has an expected running time of $2^{O(\sqrt{\delta})} n$.
These results generalise to point sets $P$ inside a hypercylinder of width $\delta$.
In this case, the factors $2^{O(\sqrt{\delta})}$ become $2^{O(\delta^{1-1/d})}$.
\end{itemize}
%\keywords{Computational geometry \and Euclidean TSP \and Bitonic TSP \and Fixed-parameter tractable algorithms}
%\subclass{68Q25 \and 68W40}
\end{abstract}

%------------------------------------------------------------------------------------------
%--------------------------------------------------------------------------------
\section{Introduction} \label{sec:intro}
%--------------------------------------------------------------------------------
In the \tspp one is given an edge-weighted complete graph and the goal is to
compute a tour---a simple cycle visiting all nodes---of minimum total
weight. Due to its practical as well as theoretical importance, the
\tspp and its many variants are among the most famous problems in computer science
and combinatorial optimization. In this paper we study the Euclidean version
of the problem. In \etsp the input is a set~$P$ of $n$ points in~$\Reals^d$,
and the goal is to compute a minimum-length tour visiting each point.
\etsp in the plane was proven to be \NP-hard in the 1970s~\cite{GareyGJ76,Papadimitriou77}.
Around the same time, Christofides~\cite{Chr76} gave an elegant (3/2)-approximation
algorithm, which works in any metric space. For a long time it was unknown if
\etsp is APX-hard, until Arora~\cite{Arora98}, and independently Mitchell~\cite{Mitchell99},
presented a PTAS. Mitchell's algorithm works for the planar case, while Arora's
algorithm also works in higher dimensions. Rao and Smith~\cite{RaoS98} later improved
the running time of Arora's PTAS, obtaining a running time of
$2^{(1/\eps)^{O(d)}} n + (1/\eps)^{O(d)} n\log n$ in~$\Reals^d$.
\medskip

We are interested in exact algorithms for \etsp. As mentioned, the problem is already \NP-hard in
the plane. Unlike the general (metric) version, however, it can be solved in
\emph{subexponential} time, that is, in time $2^{o(n)}$. In particular, Kann~\cite{Kann92} and
Hwang~\etal~\cite{HwangCL93} presented algorithms with $n^{O(\sqrt{n})}$
running time. Smith and Wormald~\cite{SmithW98} gave a subexponential
algorithm that works in any (fixed) dimension; its running time in $\Reals^d$
is~$n^{O(n^{1-1/d})}$. Very recently De Berg~\etal~\cite{bbkk-ethtsp-2018}
improved this to $2^{O(n^{1-1/d})}$, which is tight
up to constant factors in the exponent, under the Exponential-Time Hypothesis
(ETH)~\cite{ImpagliazzoP01}.

There has also been considerable research on special cases of \etsp that
are polynomial-time solvable. One example is
\btsp, where the goal is to find a shortest \emph{bitonic} tour. 
(A tour is bitonic if any vertical line crosses it
at most twice; here the points from the input set $P$ are assumed to have distinct
$x$-coordinates.) It is a classic exercise%
% in the standard textbook \emph{Introduction to Algorithms}%
~\cite{clrs-ita-09} to prove that \btsp can be solved in
$O(n^2)$ time by dynamic programming. De~Berg~\etal~\cite{bbjw-fgctsp-16}
showed how to speed this algorithm up to~$O(n\log^2 n)$.
When $P$ is in convex position,
then the convex hull of $P$ is a shortest tour and so
one can solve \etsp in $O(n \log n)$ time~\cite{bcko-cgaa-08}. 
Deineko~\etal~\cite{ddr-sctsp-94}
studied the case where the points need not all be on the convex hull;
the points inside the convex hull, however, are required to be collinear.
Their algorithm runs in $O(n^2)$ time. Deineko and Woeginger~\cite{dw-chkltsp-96}
extended this to the case where the points in the interior of the convex hull
lie on $k$ parallel lines, obtaining an~$O(n^{k+2})$ algorithm.
These results generalize earlier work by Cutler~\cite{c-esctsp-80}
and Rote~\cite{r-nlctsp-92} who consider point sets lying on three,
respectively $k$, parallel lines. Deineko~\etal~\cite{DeinekoHOW06} gave a fixed-parameter tractable algorithm for \etsp where the parameter $k$ is the number of points inside the convex hull, with running time $O(2^kk^2n)$.
Finally, Reinhold~\cite{r-srmcp-65} and Sanders~\cite{s-ec-68} proved that when there exists
a collection of disks centered at the points~in $P$ whose
intersection graph is a single cycle---this is called the
necklace condition---then the tour following the cycle is optimal.
Edelsbrunner~\etal~\cite{erw-tnc-89} gave an $O(n^2\log n)$ algorithm to verify whether such a collection of disks exists (and, if so, find one).

%------------------------------------------------------------------------------
\mypara{Our contribution.}
%------------------------------------------------------------------------------
The computational complexity of \etsp in $\Reals^d$ is $2^{\Theta(n^{1-1/d})}$
(for $d\geq 2$), assuming ETH. Thus the complexity depends heavily on the
dimension~$d$. This is most pronounced when we compare the complexity for $d=2$
with the trivial case~$d=1$: in the plane \etsp takes $2^{\Theta(\sqrt{n})}$
time in the worst case, while the 1-dimensional case is trivially solved in
$O(n\log n)$ time by sorting the points. We study the complexity of \etsp
for planar point sets that are ``almost 1-dimensional''. In particular, we
assume that the point set~$P$ is contained in the hypercylinder $\Reals \times \ball^{d-1}(\delta/2)$, where $\ball^{d-1}(\delta/2)$ denotes the closed $(d-1)$-dimensional ball with radius $\delta/2$, for some relatively small~$\delta$ and some arbitrary but fixed $d \geq 2$. We investigate how the complexity of \etsp depends
on the parameter~$\delta$. As any instance of \etsp can be scaled to fit inside a hypercylinder,
we need to make some additional restriction on the input. We consider three scenarios.
\begin{itemize}
\item \emph{Integer $x$-coordinates, with $d=2$.}
\btsp can be solved in $O(n\log^2 n)$ time~\cite{bbjw-fgctsp-16}.
It is natural to conjecture that for points with distinct integer $x$-coordinates inside a sufficiently narrow strip, an optimal bitonic tour is a shortest tour overall.
We give a (partially computer-assisted) proof that this is indeed the case: we prove that when $\delta\leq 2\sqrt{2}$ an optimal bitonic tour is optimal overall, and we show that the bound $2\sqrt{2}$ is best possible.
\item \emph{Sparse point sets.}
We generalize the case of integer $x$-coordinates to the case where the drum $[x,x+1]\times \ball^{d-1}(\delta/2)$ contains $O(1)$ points for all $x \in \Reals$.
Furthermore, we investigate how the complexity of \etsp grows with~$\delta$.
We show that for sparse point sets in $\Reals^2$ an optimal tour must be $k$-tonic---a tour is $k$-tonic if it intersects any vertical line at most $k$-times---for $k=O(\sqrt{\delta})$.
This suggests that one might be able to use a dynamic-programming algorithm similar to the ones for points on $k$ parallel lines~\cite{dw-chkltsp-96,r-nlctsp-92}.
The latter algorithms run in~$O(n^k)$ time, suggesting that a running time of $n^{O(\sqrt{\delta})}$ is achievable for $d=2$ in our case.
We give a much more efficient algorithm, which is fixed-parameter tractable (and subexponential) with respect to the parameter~$\delta$, and which generalizes to $d > 2$.
Its running time for sparse point sets in $\Reals^d$ is $2^{O(\delta^{1-1/d})}n + O(\delta^2 n^2)$.
\item \emph{Random point sets.}
In the third scenario the points in~$P$ are drawn independently and uniformly at random from the hypercylinder $[0,n]\times \ball^{d-1}(\delta/2)$.
For this case we prove that a very similar algorithm to the algorithm for sparse point sets has an expected running time of $2^{O(\delta^{1-1/d})}n$, which is linear if $\delta=O(1)$.
\end{itemize}

%-------------------------------------------------------------------
\mypara{Notation and terminology.}
%-------------------------------------------------------------------
Let $P:=\{p_1,\ldots,p_n\}$ be a set of points in a $d$-dimensional hypercylinder with radius $\delta/2$ ---we call such a hypercylinder a \emph{$\delta$-cylinder}---which we assume without loss of generality to be~$\Reals\times \ball^{d-1}(\delta/2)$.
We denote the $x$-coordinate of a point $p \in \Reals^d$ by $x(p)$,
and its other coordinates by $y_1(p),...,y_{d-1}(p)$.
To simplify the notation, we also write $x_i$ for $x(p_i)$. We sort the points in $P$ such that $x_i \leq x_{i+1}$ for all $1 \leq i < n$.
We define the \emph{spacing} $\Delta_i \mydef x_{i+1} - x_i$.

For two points $p,q \in \Reals^d$, we write $pq$ to denote the \emph{directed} edge from $p$ to $q$.
Paths are written as lists of points, so $(q_1, q_2, \ldots, q_m)$
denotes the path consisting of the edges $q_1q_2,\ldots,q_{m-1}q_m$.
All points in a path must be distinct, except possibly $q_1=q_m$ in which case
the path is a tour. The length of an edge $pq$ is denoted by~$|pq|$, and the
total length of a set $E$ of edges is denoted by~$\|E\|$.
Finally, a \emph{separator} is a hyperplane orthogonal to the $x$-axis, not containing any of the points in $P$.
%------------------------------------------------------------------------------------------

%------------------------------------------------------------------------------------------
%--------------------------------------------------------------------------------
\section{Bitonicity for points with integer \texorpdfstring{$x$}{x}-coordinates} \label{sec:bitonic}
%--------------------------------------------------------------------------------
In this section we consider the case where $d=2$ and the points in $P$ have distinct integer $x$-coordinates, which implies that $\Delta_i \in \mathbb{N}^+$ for all $i$.
Note that separators in $\Reals^2$ are vertical lines.
For simplicity, we will write $y(p) \mydef y_1(p)$, and $y_i \mydef y(p_i)$.
We will also assume the points are in the \textit{$\delta$-strip} $\Reals \times [0,\delta]$.

For our purposes, two separators~$s,s'$ that induce the same partitioning of~$P$ are equivalent.
Therefore, we can define $\sepset := \{s_1,\ldots,s_{n-1}\}$ as the set of all combinatorially distinct separators, obtained by taking one separator $s_i$ between any two points $p_i,p_{i+1}$.
Let $E$ be a set of edges with endpoints in~$P$.
The \emph{tonicity of~$E$ at a separator~$s$}, written as $\ton(E,s)$, is the number of edges in $E$ crossing~$s$.
We say that a set $E$ has \emph{lower tonicity} than a set $F$ of edges, denoted by $E \preceq F$, if $\ton(E,s_i) \leq \ton(F,s_i)$ for all $s_i\in \sepset$.
The set~$E$ has \emph{strictly lower tonicity}, denoted by $E \prec F$, if there also exists at least one $i$ for which $\ton(E,s_i) < \ton(F,s_i)$.
Finally, we call a set $E$ of edges \emph{$k$-tonic} if $\ton(E,s_i) \leq k$ for all $s_i\in \sepset$.
We will use the terms \emph{monotonic} and \emph{bitonic} to denote 1-tonic and 2-tonic, respectively.
We say $E$ has \emph{exact tonicity} $k$ if $E$ is $k$-tonic, but not $(k-1)$-tonic. 
%------------------------------------------------------------------------------------------

%------------------------------------------------------------------------------------------
%--------------------------------------------------------------------------------
The goal of this section is to prove the following theorem.
%--------------------------------------------------------------------------------
\begin{theorem} \label{thm:bitonic:main}
Let $P$ be a set of points with distinct and integer $x$-coordinates in a $\delta$-strip.
When $\delta\leq 2\sqrt{2}$, a shortest bitonic tour on $P$ is a shortest tour overall.
Moreover, for any $\delta>2\sqrt{2}$ there is a point set $P$ in a $\delta$-strip such that a shortest bitonic tour on $P$ is not a shortest tour overall.
\end{theorem}
%--------------------------------------------------------------------------------
The construction for the case $\delta>2\sqrt{2}$ is shown in Figure~\ref{fig:bitonic:2sqrt2_counterex}.
It is easily verified that, up to symmetrical solutions, the tours $T_1$ and $T_2$ are the only candidates for the shortest tour.
%--------------------------------------------------------------------------------
\begin{figure}
\begin{center}
\includegraphics{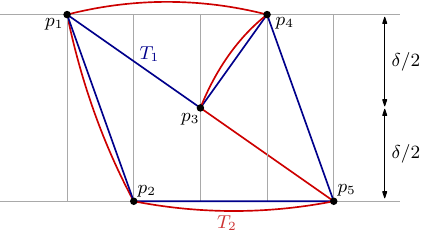}
\caption{Construction for $\delta > 2 \sqrt{2}$ for Theorem~\protect\ref{thm:bitonic:main}.
The grey vertical segments are at distance~1 from each other.
If $\delta > 2 \sqrt{2}$ then $T_1$, the shortest bitonic tour (in blue), is longer than $T_2$, the shortest non-bitonic tour (in red).}
\label{fig:bitonic:2sqrt2_counterex}
\end{center}
\end{figure}
%--------------------------------------------------------------------------------
% Both $T_1$ and $T_2$ use the edges $p_1p_2, p_2,p_5, p_3p_4$, and
% $\|p_1 p_3\|=\|p_3 p_5\|$.
Observe that
$\|T_2\| - \|T_1\| = |p_1 p_4| - |p_4 p_5| = 3 - \sqrt{1+\delta^2}$.
Hence, for $\delta>2\sqrt{2}$ we have $\|T_2\| < \|T_1\|$, which proves the lower bound of Theorem~\ref{thm:bitonic:main}.
The remainder of the section is devoted to proving the first statement.
\medskip

Let $P$ be a point set in a $\delta$-strip for $\delta=2\sqrt{2}$, where all points in $P$ have distinct integer $x$-coordinates.
Among all shortest tours on $P$, let $\opt$ be one that is minimal with respect to the $\preceq$-relation;
$\opt$ exists, since the number of different tours on $P$ is finite.
We claim that $\opt$ is bitonic, proving the upper bound of Theorem~\ref{thm:bitonic:main}.

Suppose for a contradiction that $\opt$ is not bitonic.
Let $s^*\in\sepset$ be the rightmost separator for which $\ton(\opt,s^*)>2$.
Note that for any two consecutive separators $s_i, s_{i+1}$, the difference in tonicity at those separators is at most 2.
Hence, $\ton(\opt,s^*)=4$.
Let $F$ be the four edges of $\opt$ crossing~$s^*$, and let $E$ be the remaining set of edges of~$\opt$.
Let $Q$ be the set of endpoints of the edges in~$F$.
We will argue that there exists a set~$F'$ of edges with endpoints in~$Q$ such that $E\cup F'$ is a tour and (i) $\|F'\| < \|F\|$, or (ii) $\|F'\| = \|F\|$ and $F' \prec F$.
We will call such an $F'$ \emph{superior to} $F$.
Option~(i) contradicts that $\opt$ is a shortest tour.
Since $E\cup F' \prec E\cup F$ if and only if $F' \prec F$, (ii) contradicts that $\opt$ is a shortest tour that is minimal with respect to~$\preceq$.
Hence, proving that such a superior set~$F'$ exists finishes the proof.

The remainder of the proof proceeds in two steps.
In the first step we move the points in $Q$, obtaining a set $\overline{Q}$ with consecutive integer coordinates and an edge set $\ovf$.
These will be such that if an $\ovf'$ superior to $\ovf$ exists, then there also exists an $F'$ superior to $F$.
In the second step we then give a computer-assisted proof that the desired set~$\ovf'$ exists.

%--------------------------------------------------------------------------------
\mypara{Step~1: Finding a suitable~$\overline{Q}$ with consecutive $x$-coordinates.}
%--------------------------------------------------------------------------------
Let $\opt$, $s^*$, $E$, $F$ and $Q$ be defined as above. We assume without loss of generality that the $x$-coordinate of $s^*$ is equal to $x^*+\frac{1}{2}$,
where~$x^*$ is the largest integer such that the line $x=x^*+\frac{1}{2}$ intersects all four edges in~$F$.
Since the actual edges in $E$ are not important for our arguments, we replace them by abstract ``connections'' specifying which pairs of endpoints of the edges in~$F$ are connected by paths of edges in~$E$.
It will be convenient to duplicate the points in~$Q$ that are shared endpoints of two edges in $F$, and add a connection between the two copies; see Figure~\ref{fig:bitonic:connectivity_pattern}.
We denote the set of connections obtained in this way by $\widetilde{E}$, and we call $\widetilde{E}$ the \emph{connectivity pattern} of~$F$ (in $E\cup F$).
%--------------------------------------------------------------------------------
\begin{figure}
\begin{center}
\includegraphics{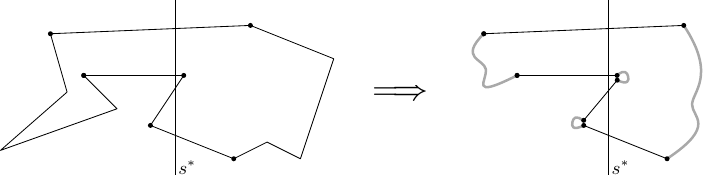}
\caption{Replacing the paths connecting endpoints of edges in $F$ by abstract connections.
    The copies of duplicated shared endpoints are slightly displaced in the figure to be able to distinguish them, but they are actually coinciding.}
\label{fig:bitonic:connectivity_pattern}
\end{center}
\end{figure}
%--------------------------------------------------------------------------------
\medskip

Next we show how to move the points in~$Q$ such that the modified set~$\overline{Q}$ uses consecutive $x$-coordinates.
Recall that $s^* : x=x^* +\frac{1}{2}$ is a separator that intersects all edges in~$F$.
Let $Q_{\myleft}$ and $Q_{\myright}$ be the subsets of points from~$Q$ lying to the left and right of $s^*$, respectively.
We will move the points in $Q_{\myleft}$ such that they will get consecutive $x$-coordinates with the largest one being equal to~$x^*$, while the points in $Q_{\myright}$ will get consecutive $x$-coordinates with the smallest one being~$x^*+1$.

We move the points in $Q_{\myleft}$ as follows.
Let $z\leq x^*$ be the largest $x$-coordinate currently not in use by any of the points in $Q_{\myleft}$.
If $Q_{\myleft}$ lies completely to the right of the line $\ell(z):x=z$, then we are done:
the set of $x$-coordinates used by points in $Q_{\myleft}$ is $\{z+1,\ldots,x^*\}$.
Otherwise, we take any point to the left of~$\ell(z)$, and we move it along its edge $e \in F$ to the point $e\cap \ell(z)$;
see Figure~\ref{fig:bitonic:reduction_to_consecutive}(i).
This process is repeated until the points of $Q_{\myleft}$ have consecutive $x$-coordinates.
Note that duplicate $x$-coordinates caused by duplicating shared endpoints can still exist, if neither of the duplicates was moved.

After moving the points in $Q_{\myleft}$ we treat $Q_{\myright}$ in a similar manner;
the only difference is that now we define $z> x^*$ to be the smallest $x$-coordinate currently not in use by any of the points in $Q_{\myright}$.
Figure~\ref{fig:bitonic:reduction_to_consecutive}(ii) shows the final result for the example in part~(i) of the figure.
%--------------------------------------------------------------------------------
\begin{figure}
\begin{center}
\includegraphics{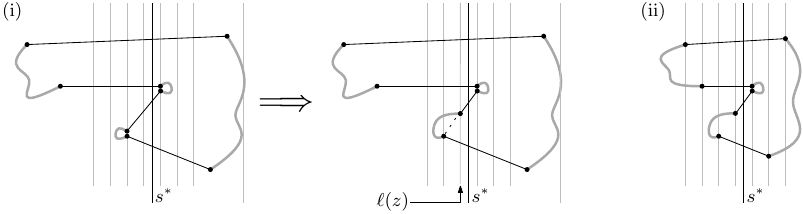}
\caption{The process of moving the points in~$Q$. Grey vertical lines have integer~$x$-coordinates.
         (i) Moving a point in $Q_{\myleft}$ so that it gets $x$-coordinate~$z$.
         (ii) A possible configuration after $Q_{\myleft}$
              and $Q_{\myright}$  have been treated.}
\label{fig:bitonic:reduction_to_consecutive}
\end{center}
\end{figure}
%--------------------------------------------------------------------------------

Before we prove that this procedure preserves the desired properties, two remarks are in order about the process described above.
First, in each iteration we may have different choices for the edge~$e$ crossing $\ell(z)$, and the final result depends on these choices.
Second, when we move a point in $Q$ to a new location, then the new $x$-coordinate is not used by $Q$ but it may already be used by points in $P\setminus Q$.
Neither of these facts causes any problems for the coming arguments.
\medskip

Let $\overline{Q}$ be the set of points from~$Q$ after they have been moved to their new locations, and let $\ovf$ be the set of edges from~$F$ after the move.
With a slight abuse of notation we still use $\widetilde{E}$ to specify the connectivity pattern on $\ovf$, which is simply carried over from~$F$.
The following lemma shows that we can use $\ovf$, $\overline{Q}$ and $\widetilde{E}$ in Step~2 of the proof.

%--------------------------------------------------------------------------------
\begin{lemma} \label{lem:bitonic:reduction_to_consecutive}
Let $E, F, Q, \widetilde{E}, \ovf,\overline{Q}$ be defined as above.
Let $\ovf'$ be any set of edges (with endpoints in $\overline{Q}$) superior to $\ovf$, such that $\widetilde{E} \cup \ovf'$ is a tour.
Then there is a set of edges~$F'$ (with endpoints in $Q$) superior to $F$ such that $E \cup F'$ is a tour.
\end{lemma}
%--------------------------------------------------------------------------------
\begin{proof}
We define $F'$ in the obvious way, by simply taking $\ovf'$ and replacing each endpoint (which is a point in $\overline{Q}$) by the corresponding point in~$Q$.
Clearly $E \cup F'$ forms a tour if $\widetilde{E} \cup \ovf'$ forms a tour.

Suppose $\ovf'$ is superior to $\ovf$.
We will show that $F'$ is superior to $F$.
We will do this by first proving that $\|F\| - \|F' \| \geq \|\ovf\| - \|\ovf'\|$.
Then, we prove that $F' \preceq F$.
Finally, we prove that if there exists a separator $\overline{s} \in \sepset$ such that $\ton(\ovf',\overline{s}) < \ton(\ovf,\overline{s})$, then $\ton(F',s^*) < \ton(F,s^*)$. \medskip

Recall that each edge~$\overline{e}\in\ovf$ is obtained from the corresponding edge~$e\in F$ by moving one or both endpoints along the edge~$e$ itself.
Also recall that we duplicated shared endpoints of edges in $F$, so if we move a point in~$Q$, then we move the endpoint of a single edge in~$F$.
Hence,
$$\|F\| - \|\ovf\| = \mbox{total distance over which the points in $Q$ are moved}.$$
Since we added a connection to $\widetilde{E}$ between the two copies of a shared endpoint, each point in~$Q$ is incident to exactly one connection in~$\widetilde{E}$ and, hence, to exactly one edge in~$F'$.
This means that if we move a point in $Q$, then we move the endpoint of a single edge in $F'$, so
$$\|F'\| - \|\ovf'\| \leq \mbox{total distance over which the points in $Q$ are moved}.$$
We conclude that $\|F\| - \|F' \| \geq \|\ovf\| - \|\ovf'\|$. \medskip

Next, we claim that $F' \preceq F$.
We will prove this by showing that $F$ is maximal with respect to $\preceq$.
Let $s$ be any separator.
Let $G$ be any edge set on $Q$ such that $\widetilde{E} \cup G$ forms a tour.
Now, the tonicity of $G$ at $s$ is bounded by the minimum of the number of points of $Q$ to the left of $s$ and the number of points of $Q$ to the right of $s$.
Note that $F$ attains this bound for every separator $s$.
We conclude that $F$ is maximal with respect to $\preceq$, and therefore, $F' \preceq F$ indeed holds. \medskip

Finally, let $\overline{s} \in \sepset$ be such that $\ton(\ovf',\overline{s}) < \ton(\ovf,\overline{s})$.
We will show that $\ton(F',s^*) < \ton(F,s^*)$.
Without loss of generality, let $\overline{s}$ either be to the left of $s^*$, or $s^*$ itself.
Since $\ton(\ovf',\overline{s}) < \ton(\ovf,\overline{s})$, the edge set $\ovf'$ contains an edge $e$ fully to the left of $\overline{s}$: otherwise, the tonicity of $\ovf'$ at $\overline{s}$ would be equal to the number of points to the left of $\overline{s}$, which is equal to the tonicity of $\ovf$ at $\overline{s}$.
Furthermore, since
$$\ton(F,s^*) = \ton(\ovf,s^*)$$
and
$$\ton(F',s^*) = \ton(\ovf',s^*),$$
we conclude that $\ton(F',s^*) < \ton(F,s^*)$. \medskip

Recall that by the definition of superior, $F'$ is superior to $F$ if and only if (a) $\|F'\| < \|F\|$, or (b) $\|F'\| = \|F\|$ and $F' \prec F$.
We have proven that (i) $\|F\| - \|F' \| \geq \|\ovf\| - \|\ovf'\|$, (ii) $F' \preceq F$, and (iii) if $\ton(\ovf',s) < \ton(\ovf,s)$ for some separator $s \in \sepset$, then $\ton(F',s^*) < \ton(F,s^*)$.

Suppose that $\ovf'$ is superior to $\ovf$.
If $\|\ovf'\| < \|\ovf\|$, then by (i), we get that (a) holds.
Else, $\|\ovf'\| = \|\ovf\|$, so by (i), we get that either (a) holds, or $\|F'\| = \|F\|$, satisfying the first requirement of (b).
Furthermore, since in this case $\ovf' \prec \ovf$, there exists a separator $\overline{s}$ satisfying $\ton(\ovf',\overline{s}) < \ton(\ovf,\overline{s})$.
By (iii), we have $\ton(F',s^*) < \ton(F,s^*)$.
Combining this with (ii) gives us that $F' \prec F$, which is the second requirement of (b).

We conclude that if $\ovf'$ is superior to $\ovf$, then $F'$ is superior to $F$.
\end{proof}
%--------------------------------------------------------------------------------

%--------------------------------------------------------------------------------
\mypara{Step~2: Finding the set~$F'$.}
%--------------------------------------------------------------------------------
The goal of Step~2 of the proof is the following:
given the tour~$\opt = E\cup F$ inside a $\delta$-strip of width~$\delta=2\sqrt{2}$, show that there exists a set~$F'$ of edges such that $E\cup F'$ is a tour and $F'$ is superior to $F$.
Lemma~\ref{lem:bitonic:reduction_to_consecutive} implies that we may work with $\widetilde{E}$ and $\ovf$ instead of $E$ and $F$ (and then find $\ovf'$ instead of~$F'$).

In Step~1 we duplicated shared endpoints of edges in $E$.
We now merge these two copies again if they are still at the same location.
This will always be the case for the shared endpoint immediately to the right of the separator~$s^*$:
we picked $s^*: x=x^*+\frac{1}{2}$ such that there is a shared endpoint at $x=x^*+1$ (since $s^*$ is the rightmost separator with $\ton(\opt,s^*) = 4$), and the copies of this endpoint will not be moved.
So if $n_{\myleft}$ and $n_{\myright}$ denote the number of distinct endpoints to the right and left of~$s^*$, respectively, then $n_{\myright}\in\{2,3\}$ and $n_{\myleft}\in\{2,3,4\}$.
We thus have six cases in total for the pair~$(n_{\myleft},n_{\myright})$, as depicted in Figure~\ref{fig:bitonic:overlineF_cases}.
%--------------------------------------------------------------------------------
\begin{figure}
\begin{center}
\includegraphics{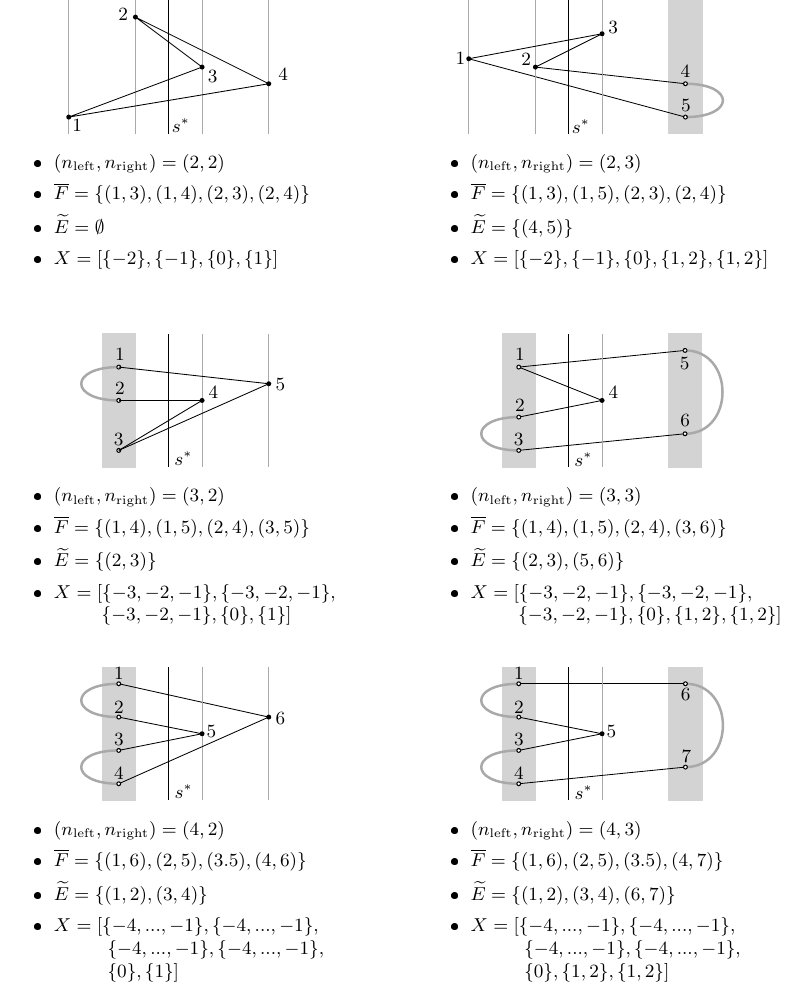}
\caption{The six different cases that result after applying Step~1 of the proof.
Points indicated by filled disks have a fixed $x$-coordinate.
The left-to-right order of points drawn inside a grey rectangle, on the other hand, is not known yet.
The vertical order of the edges is also not fixed, as the points can have any $y$-coordinate in the range~$[0,2\sqrt{2}]$.}
\label{fig:bitonic:overlineF_cases}
\end{center}
\end{figure}
%--------------------------------------------------------------------------------
Each of the six cases has several subcases, depending on the left-to-right order of the vertices inside the gray rectangles in the figure.
For any fixed ordering we can still vary the $y$-coordinates in the range~$[0,\delta]$.
This may lead to scenarios where different sets~$\ovf'$ are required.
We handle this potentially huge amount of cases in a computer-assisted manner, using an automated prover \emph{FindShorterTour}$(n_\myleft,n_\myright,\ovf,\widetilde{E},X,\delta,\eps)$.
The input parameter~$X$ is an array where $X[i]$ specifies the set from which the $x$-coordinate of the $i$-th point in the given scenario may be chosen.
Here, we assume w.l.o.g.~that $x(s^*) = -1/2$;
see Figure~\ref{fig:bitonic:overlineF_cases}.
The smaller the precision parameter~$\eps$ is, the more precise the output will be. Its precise role will be explained below.

The output of \emph{FindShorterTour} is a list of \emph{scenarios} and an \emph{outcome} for each scenario.
A scenario contains for each point~$q$ an $x$-coordinate~$x(q)$ from the set of allowed $x$-coordinates for~$q$, and a range~$\yrange(q)\subseteq [0,2\sqrt{2}]$ for its $y$-coordinate.
This $y$-range is an interval of length at least~$\eps/2$.
The outcome of a scenario is either \success or \fail.
An outcome \success means that a set $\ovf'$ has been found with the desired properties:
$\widetilde{E}\cup \ovf'$ is a tour, and for all possible instantiations of the scenario---that is, all choices of $y$-coordinates from the $y$-ranges in the scenario---we have $\|\ovf'\| < \|\ovf\|$.
An outcome \fail means that such an $\overline{F'}$ has not been found.
It does not guarantee that such an $\ovf'$ does not exist for this scenario.

The list of scenarios is complete in the sense that for any instantiation of the input case there is a scenario that covers it.

\emph{FindShorterTour} works brute-force.
It checks all possible combinations of $x$-coordinates and subdivides the $y$-coordinate ranges, until a suitable $\ovf'$ can be found or until the $y$-ranges have length smaller than the precision parameter~$\eps$.
The implementation details of the procedure are in Appendix~\ref{sec:app_prover}.
\medskip

Note that case $(n_\myleft,n_\myright)=(2,3)$ in Figure~\ref{fig:bitonic:overlineF_cases}
is a subcase of case $(n_\myleft,n_\myright)=(3,2)$, if we exchange the roles of the points lying to the left and to the right of~$s^*$, swap the labels of points 5 and 1, and swap the labels of points 2 and 4.
Hence, we ignore this subcase and run our automated prover on the remaining five cases, where we set $\eps:= 0.001$.
It successfully proves the existence of a suitable set~$\ovf'$ in four cases;
the case where the prover fails is the case $(n_\myleft,n_\myright)=(3,2)$.
For this case, it fails for the two scenarios depicted in Figure~\ref{fig:bitonic:worstcase}.
All other scenarios for these cases are handled successfully (up to symmetries).
%--------------------------------------------------------------------------------
\begin{figure}
\begin{center}
\includegraphics{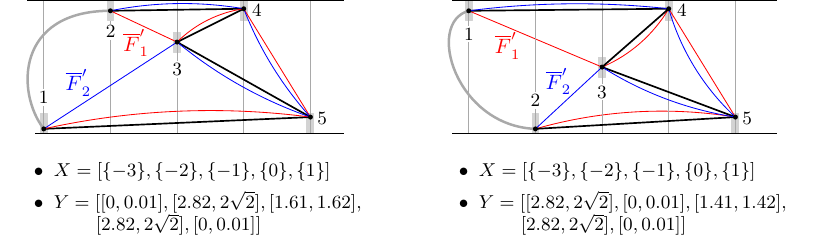}
\caption{Two scenarios covering all subscenarios where the automated prover fails, up to symmetries.
Each point has a fixed $x$-coordinate and a $y$-range specified by the array~$Y$;
the resulting possible locations are shown as small grey rectangles (drawn larger than they actually are for visibility).
For all subscenarios, at least one of $\ovf'_1$ (in red) and $\ovf'_2$ (in blue) is at most as long as $\ovf$ (in black).}
\label{fig:bitonic:worstcase}
\end{center}
\end{figure}
%--------------------------------------------------------------------------------
For both scenarios in Figure~\ref{fig:bitonic:worstcase} we consider two alternatives for the set $\ovf'$: the set $\ovf'_1$ shown in red in Figure~\ref{fig:bitonic:worstcase}, and the set $\ovf'_2$ shown in blue.
We will show that in any instantiation of both scenarios, either $\ovf'_1$ or $\ovf'_2$ is at least as short as $\ovf$.
Since both alternatives are bitonic, this finishes the proof.

For $1\leq i\leq 5$, let $q_i$ be the point labeled~$i$ in the left scenario in Figure~\ref{fig:bitonic:worstcase}.
We first argue that we can assume without loss of generality that $y(q_2)=y(q_4)=2\sqrt{2}$ and $y(q_1)=y(q_5)=0$.
To this end, consider an arbitrary instantiation of this scenario, and imagine moving $q_2$ and $q_4$ up to the line~$y=2\sqrt{2}$, and moving $q_1$ and $q_5$ down to the line~$y=0$.
It suffices to show, for $i\in\{1,2\}$, that if we have $\|\ovf'_i\| \leq \|\ovf\|$ after the move, then we also have $\|\ovf_i\| \leq \|\ovf\|$ before the move.
To see why this indeed holds, we will need the following observation.

\begin{observation}\label{obs:bitonic:moving}
Let $p$ be a point in $\Reals^2$.
Suppose we move $p$ in direction $\vec{v}$ at unit speed.
Let $p(t)$ denote the position of point $p$ at time $t$.
Let $\vec{p}(t)$ be the corresponding vector.
Then $\frac{d}{dt} |\vec{p}(t)| = \frac{\vec{p}(t)\cdot v}{|\vec{p}(t)|}$, where $`\,\cdot$' denotes the inner product.
This also equals the cosine of the angle between $\vec{p}(t)$ and $\vec{v}$.
\end{observation}
\begin{proof}
Assume without loss of generality that $\vec{v} = (1,0)$.
Then
\begin{align*}
\frac{d}{dt} |\vec{p}(t)| &= \frac{d}{dt} \sqrt{x(p(t))^2+y(p(t))^2}\\
&= \frac{d}{dt} \sqrt{(x(p)+t)^2+y(p)^2}\\
&= \frac{x(p)+t}{\sqrt{(x(p)+t)^2+y(p)^2}}\\
&= \frac{\vec{p}(t)\cdot \vec{v}}{|\vec{p}(t)|}.
\end{align*}
Finally, for any vector $\vec{v}_1,\vec{v}_2$ we have $\vec{v}_1 \cdot \vec{v}_2 = |\vec{v}_1| |\vec{v}_2| \cos \alpha$, where $\alpha$ is the angle between $\vec{v}_1$ and $\vec{v}_2$.
Since $\vec{v}$ is a unit-length vector, this finalizes our argument.
\end{proof}

This leads us to the following corollary.
%--------------------------------------------------------------------------------
\begin{corollary}\label{cor:bitonic:moving}
Let $a,b,c$ be three points.
Let $\ell$ be the vertical line through $c$, and suppose we move $c$ downwards along $\ell$.
Let $\alpha$ be the smaller angle between~$ac$ and~$\ell$ if $y(c) < y(a)$, and the larger angle otherwise.
Let $\beta$ be the smaller angle between~$bc$ and~$\ell$ if $y(c) < y(b)$, and the larger angle otherwise.
Suppose $\alpha<\beta$ throughout the move.
Then the move increases $|ac|$ more than it increases~$|bc|$.
\end{corollary}
\begin{proof}
From Observation~\ref{obs:bitonic:moving} we can directly conclude that at any time during the move, the derivative of $|ac|$ is strictly larger than the derivative of $|bc|$.
\end{proof}
%--------------------------------------------------------------------------------
We can now repeatedly apply Corollary~\ref{cor:bitonic:moving} to our current scenario, the left scenario in Figure~\ref{fig:bitonic:worstcase}. Let us first compare $\|\ovf'_1\|$ to $\|\ovf\|$. Moving $q_1$ downwards to $y = 0$ changes both by the same amount.
We can apply Corollary~\ref{cor:bitonic:moving} to $q_2$: by moving it upwards, $|q_2 q_3|$ increases more than $|q_2 q_4|$, thus $\|\ovf'_1\|-\|\ovf\|$ increases. Then, moving $q_4$ upwards, $|q_2 q_4|$ decreases while $|q_4 q_5|$ increases, and therefore $\|\ovf'_1\|-\|\ovf\|$ increases once more. Finally, we can apply Corollary~\ref{cor:bitonic:moving} to show that if we move $q_5$ downwards, $|q_4 q_5|$ increases more than $|q_3 q_5|$, and therefore increasing $\|\ovf'_1\|-\|\ovf\|$ again.
Since none of the moves decreased $\|\ovf'_1\|-\|\ovf\|$, we conclude that if $\|\ovf'_1\| \leq \|\ovf\|$ after the move, then we also have $\|\ovf_1\| \leq \|\ovf\|$ before the move. We can prove analogously that if $\|\ovf'_2\| \leq \|\ovf\|$ after the move, then we also have $\|\ovf_2\| \leq \|\ovf\|$ before the move.

So now assume $y(q_2)=y(q_4)=2\sqrt{2}$ and $y(q_1)=y(q_5)=0$.
Let $y := y(q_3)$.
If $y\geq (8\sqrt{2})/7$ then
\[
|q_2 q_3|+|q_4 q_5| = \sqrt{1+(2\sqrt{2}-y)^2} + 3 \leq 2 + \sqrt{4+y^2} = |q_2 q_4|+|q_3 q_5|,
\]
so $\| \ovf'_1 \| \leq \| \ovf \|$.
On the other hand, if $y\leq (8\sqrt{2})/7$ then
\[
|q_1 q_3|+|q_4 q_5| = \sqrt{4+y^2} + 3   \leq   \sqrt{1+(2\sqrt{2}-y)^2} + 4 = |q_3 q_4|+|q_1 q_5|,
\]
so $\| \ovf'_2 \| \leq \| \ovf \|$.
So either $\ovf'_1$ or $\ovf'_2$ is at least as short as $\ovf$,
finishing the proof for the left scenario in Figure~\ref{fig:bitonic:worstcase}.
The proof for the right scenario in Figure~\ref{fig:bitonic:worstcase} is analogous, with cases $y\geq \sqrt{2}$ and $y\leq \sqrt{2}$.
This finishes the proof for the right scenario
and, hence, for Theorem~\ref{thm:bitonic:main}.
%------------------------------------------------------------------------------------------

%------------------------------------------------------------------------------------------
%------------------------------------------------------------------------
\section{The tonicity of TSP tours of sparse point sets} 
\label{sec:2ktonic}
%------------------------------------------------------------------------
In this section we investigate the tonicity of optimal TSP tours for planar sets $P$ that are sparse.
Recall that $P$ is sparse if for any $x \in \Reals$ the drum $[x,x+1]\times \ball_{d-1}(\delta/2)$ contains $O(1)$ points.
We start with an easy observation that holds for point sets in an arbitrary number of dimensions, illustrated in Fig.~\ref{fig:2ktonicproofs:switcheroo}.
%------------------------------------------------------------------------
\begin{figure}
\begin{center}
\includegraphics[]{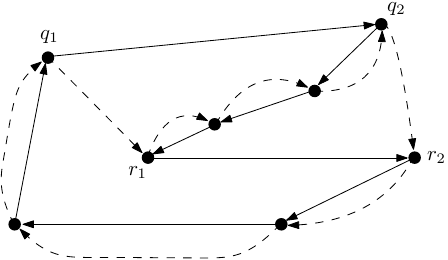}
\caption{Illustration of Observation~\ref{obs:2ktonicproofs:switcheroo}.
$T$ is the solid line, $T'$ is the dashed line.
Between $x(r_1)$ and $x(q_2)$, the tonicity of $T'$ is 2 lower than the tonicity of $T$.}
\label{fig:2ktonicproofs:switcheroo}
\end{center}
\end{figure}
%------------------------------------------------------------------------
\begin{observation}\label{obs:2ktonicproofs:switcheroo}
Let $T$ be a tour containing the (directed) edges $q_1 q_2$ and $r_1 r_2$.
Let $|q_1 r_1|+|q_2 r_2| \leq |q_1 q_2|+|r_1 r_2|$.
Then if we swap the edges $q_1 q_2$ and $r_1 r_2$ for the edges $q_1 r_1$ and $q_2 r_2$ we obtain a tour $T'$ with $\|T'\| \leq \|T\|$.
Similarly, if $|q_1 r_1|+|q_2 r_2| < |q_1 q_2|+|r_1 r_2|$ then $\|T'\| < \|T\|$.
Furthermore, if $q_1$ and $r_1$ are to the left of $q_2$ and $r_2$, then $T' \prec T$.
\end{observation}
%------------------------------------------------------------------------
\begin{proof}
The fact that $\|T'\| \leq \|T\|$ (or $\|T'\| < \|T\|$) is trivial.
Furthermore, if $q_1$ and $r_1$ are to the left of $q_2$ and $r_2$, then $T' \prec T$ holds:
between $\max(x(q_1),x(r_1))$ and $\min(x(q_2),x(r_2))$ the tonicity has been lowered by~2.
Everywhere else the tonicity remains unchanged, so $T' \prec T$.
\end{proof}
%------------------------------------------------------------------------
We will need the following observation, illustrated in Fig.~\ref{fig:2ktonicproofs:linesum}.
\begin{figure}
\begin{center}
\includegraphics{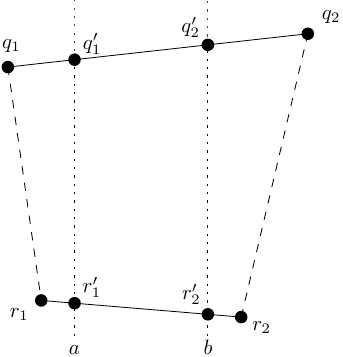}
\caption{Illustration of Observation~\ref{obs:2ktonicproofs:linesum}.
Since $|q_1q_2| + |r_1r_2| < |q_1r_1|+|q_2r_2|$ holds, $|q_1'q_2'|+|r_1'r_2'| < |q_1'r_1'| + |q_2'r_2'|$ also holds.}
\label{fig:2ktonicproofs:linesum}
\end{center}
\end{figure}
%------------------------------------------------------------------------
\begin{observation}\label{obs:2ktonicproofs:linesum}
Let $a,b\in\Reals$.
Let $q_1,q_2,r_1,r_2$ be any four points of $P$ with $\max(x(q_1), x(r_1))$ $\leq a < b \leq \min(x(q_2), x(r_2))$ and $|q_1q_2| + |r_1r_2| < |q_1r_1|+|q_2r_2|$. 
Let $q_1', q_2'$ be the points on $q_1q_2$ with $x$-coordinates $a, b$ respectively. 
Let $r_1', r_2'$ be the points on $r_1r_2$ with $x$-coordinates $a, b$ respectively. 
Then $2(b-a) \leq |q_1'q_2'|+|r_1'r_2'| < |q_1'r_1'| + |q_2'r_2'|$.
\end{observation}
%------------------------------------------------------------------------
\begin{proof}
When we move $q_1$ straight towards $q_2$ until it reaches $q'_1$, then $|q_1r_1|$ cannot decrease more than $|q_1q_2|$.
Therefore, $|q_1r_1| - |q_1'r_1| \leq |q_1q_2| - |q_1'q_2|$.
Combining this with our assumption that $|q_1q_2| + |r_1r_2| < |q_1r_1|+|q_2r_2|$, we get that
$|q_1'q_2| + |r_1r_2| < |q_1'r_1|+|q_2r_2|$.
Note how we have essentially swapped $q_1$ for $q_1'$.
We can repeat this process three more times.
First we swap $q_2$ for $q'_2$, then we swap $r_1$ for $r_1'$, and finally we swap $r_2$ for $r_2'$. We thus obtain that $|q_1'q_2'| + |r_1'r_2'| < |q_1'r_1'|+|q_2'r_2'|$.
We complete our proof by noting that since $x(q_1') = x(r_1') = a$ and $x(q_2') = x(r_2') = b$, we trivially have that $2(b-a) \leq |q_1' q_2'| + |r_1'r_2'|$.
%}{The lemma follows by repeating this argument for the other three points.}
\end{proof}
%------------------------------------------------------------------------
From now on, we will once more assume that $d=2$, and that the point set $P$ lies in the $\delta$-strip $\Reals \times [0,\delta]$.
Recall that $s_i$ denotes the separator between points~$p_i,p_{i+1}$, and recall the definition $\Delta_i \mydef x_{i+1} - x_i$. 
A direct consequence of Observation~\ref{obs:2ktonicproofs:linesum} is the following.
%------------------------------------------------------------------------
\begin{lemma}\label{lem:2ktonicproofs:gapseparator}
Let $k \in \mathbb{N}$ be such that $\delta \leq k \Delta_i$. Let $\opt$ be a shortest tour on~$P$.
Then there exists a shortest tour $\opt' \preceq \opt$ with $\ton(\opt',s_i) \leq 2k$.
\end{lemma}
%------------------------------------------------------------------------
\begin{proof}
Let $\opt$ be a shortest tour.
If $\ton(\opt,s_i)>2k$, then $\opt$ has at least $k+1$ edges crossing~$s_i$ from left to right.
We claim that at least one pair of edges $(q_1q_2, r_1r_2)$ has the property that $|q_1q_2|+|r_1r_2| \geq |q_1r_1|+|q_2r_2|$. 
To see this, suppose such a pair does not exist.
If we order the edges crossing~$s_i$ on where they cross~$s_i$, and then apply Observation~\ref{obs:2ktonicproofs:linesum} to each consecutive pair, we get that $2 \delta > 2 k \Delta_i$.
Specifically, let $e_1, ..., e_{k+1}$ be such a set of $k+1$ edges crossing $s_i$.
Let $e_j = (e_{j1}, e_{j2})$ for all $j$, and let $e_{j1}'$ and $e_{j2}'$ be the points on $e_j$ with $x$-coordinates $x_i$ and $x_{i+1}$, respectively.
Note that since $e_j$ crosses $s_i$ and is defined by two points of $P$, $e_{j1}'$ and $e_{j2}'$ are indeed well-defined.
Now, from Observation~\ref{obs:2ktonicproofs:linesum}, we get for all $1 \leq j \leq k$ that $2(x_{i+1}-x_i) < |e_{j1}' e_{(j+1)1}'| + |e_{j2}' e_{(j+1)2}'|$.
See Figure~\ref{fig:2ktonicproofs:gapseparator} for an example.
Therefore,
\[2k \Delta_i = \sum_{j=1}^k 2(x_{i+1}-x_i) < \sum_{j=1}^k |e_{j1}' e_{(j+1)1}'| + |e_{j2}' e_{(j+1)2}'| = |e_{11}' e_{(k+1)1}'| + |e_{12}' e_{(k+1)2}'| \leq 2 \delta.\]
This, however, directly contradicts our assumption that $\delta \leq k \Delta_i$. 
Therefore, such a pair of edges must exist.
Using this pair, we can find a tour~$\opt^* \prec \opt$ with $\|\opt^*\| \leq \|\opt\|$ (and therefore a shortest tour), by Observation~\ref{obs:2ktonicproofs:switcheroo}.
We apply this repeatedly until we have obtained a shortest tour $\opt' \preceq \opt$ with $\ton(\opt',s_i) \leq 2k$.
\begin{figure}
\begin{center}
\includegraphics[]{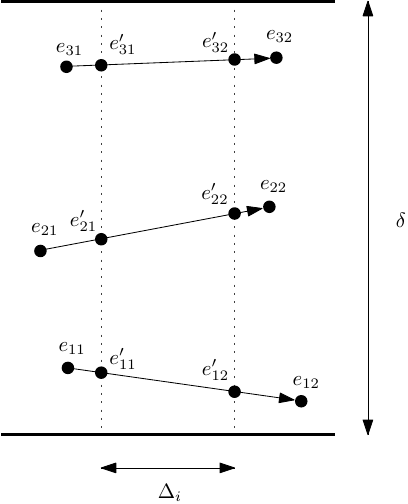}
\caption{ An example of Lemma~\ref{lem:2ktonicproofs:gapseparator}, with $k = 2$.
Suppose $|e_1|+|e_2| < |e_{11}e_{21}| + |e_{12}e_{22}|$ and $|e_2|+|e_3| < |e_{21}e_{31}|+|e_{22}e_{32}|$.
Since $\delta \geq |e_{11}'e_{21}'| + |e_{21}'e_{31}'|$ and $\delta \geq |e_{12}'e_{22}'| + |e_{22}'e_{32}'|$, we get that $2 \delta \geq |e_{11}'e_{21}'| + |e_{21}'e_{31}'| + |e_{12}'e_{22}'| + |e_{22}'e_{32}'| > |e_{11}'e_{12}'| + 2 |e_{21}'e_{22}'| + |e_{31}'e_{32}'|  \geq 4 \Delta_i$.}
\label{fig:2ktonicproofs:gapseparator}
\end{center}
\end{figure}
\end{proof}
%------------------------------------------------------------------------
A more general version is the following:
%------------------------------------------------------------------------
\begin{lemma}\label{lem:2ktonicproofs:gapseparator2}
Let $k \in \mathbb{N}$ be such that $\delta \leq (x_j-x_i)(k-j+i+1)$ for some $1 \leq i < j \leq n$. 
Let $\opt$ be a shortest tour.
Then there exists a shortest tour $\opt' \preceq \opt$ with $\ton(\opt',s_i) \leq 2k$ and $\ton(\opt',s_{j-1}) \leq 2k$.
\end{lemma}
%------------------------------------------------------------------------
\begin{proof}
Let $\opt$ be a shortest tour.
If $\ton(\opt,s_i)>2k$, then $\opt$ has at least $k+1$ edges crossing $s_i$ from left to right.
We claim that at least one pair $(q_1q_2, r_1r_2)$ of crossing edges has the property that $|q_1q_2|+|r_1r_2| \geq |q_1r_1|+|q_2r_2|$.
To see this, suppose such a pair does not exist.
Since at most $j-i-1$ of these $k+1$ edges have an endpoint in $\{p_{i+1},...,p_{j-1}\}$, at least $k+1-(j-i-1)$ of these have length at least $x_j-x_i$.
Analogously to the previous lemma, we order these edges depending on where they cross $s_i$, and then apply Observation~\ref{obs:2ktonicproofs:linesum} to each consecutive pair.
This gives us that $2 \delta > 2 (k-(j-i-1)) (x_j-x_i)$. 
This, however, directly contradicts our assumption. Therefore, such a pair of edges must exist.

Analogously to the previous lemma, we can use this to find a tour~$\opteen \prec \opt$ with $\|\opteen\| \leq \|\opt\|$ (and therefore a shortest tour), by Observation~\ref{obs:2ktonicproofs:switcheroo}.
We apply this repeatedly until we obtain a shortest tour $\opt^* \preceq \opt$ with $\ton(\opt^*,s_i) \leq 2k$. 

By symmetry, we can then find a shortest tour $\opt' \preceq \opt^*$ with $\ton(\opt',s_{j-1}) \leq 2k$.
Since $\opt' \preceq \opt^*$, we also have $\ton(\opt',s_i) \leq 2k$.
Therefore, $\opt' \preceq \opt$ is indeed a shortest tour with $\ton(\opt',s_i) \leq 2k$ and $\ton(\opt',s_{j-1}) \leq 2k$.
\end{proof}
%------------------------------------------------------------------------
We can now bound the tonicity of an optimal TSP tour for~$P$.
%------------------------------------------------------------------------
\begin{theorem}\label{thm:2ktonicproofs:main2k_sparse}
Let $P$ be a set of points inside a $\delta$-strip such that for any $x \in \Reals$ there are at most $c$ points with $x$-coordinates in the interval $[x, x+1]$.
Then there exists an optimal TSP tour on $P$ that is $2k$-tonic for $k \mydef \floori{2\sqrt{c\delta}+2c}$.
\end{theorem}
%------------------------------------------------------------------------
\begin{proof}
Let $\opt$ be a shortest tour.
We define $m \mydef \floori{\sqrt{\delta/c}+2}$.
We split the proof into two cases.
\begin{itemize}
\item \textbf{Case I: $n < 2cm$.}
Trivially, $\opt$ is $n$-tonic.
We have 
$$n < 2cm = 2 c \floor{\sqrt{\delta/c}+2} < 2 \floor{2\sqrt{c\delta}+2c} = 2 k,$$
so $\opt$ is indeed $2k$-tonic.
\item \textbf{Case II: $n \geq 2cm$.}
For any $i, j$ such that $1 \leq i$ and $j = i + c m \leq n$, we would like to apply Lemma~\ref{lem:2ktonicproofs:gapseparator2}.
To do so, we first argue that $\delta \leq (x_j - x_i) (k-j+i+1)$.
Note that since for any $x \in \Reals$ there are at most $c$ points with $x$-coordinates in the range $[x,x+1]$, we have $x_j-x_i \geq \floori{(j-i)/c}$.
Therefore, we get
\begin{align*}
(x_j - x_i) (k-j+i+1) & \geq \br{\floor{\frac{j-i}{c}}-1}(k+1-(j-i)) \\
	& = \br{\floor{\sqrt{\delta/c}+2}-1}\br{\floor{2\sqrt{c\delta}+2c}+1-c \floor{\sqrt{\delta/c}+2}} \\
	& \geq \sqrt{\delta/c}\br{2\sqrt{c\delta}+2c-c \br{\sqrt{\delta/c}+2}}\\
	& = \delta.
\end{align*}
This is true independent of our choice of $i$.
Therefore, we can first apply Lemma~\ref{lem:2ktonicproofs:gapseparator2} with $i=1$, giving us an optimal tour $\opteen \preceq \opt$ that is $2k$-tonic at $s_1$ and at~$s_{c m}$. 
When we apply this lemma again on $\opteen$ with $i=2$, we get a shortest tour $\opttwee \preceq \opteen$ that is $2k$-tonic at $s_1, s_2, s_{cm}$ and $s_{cm+1}$.
After doing this for all $i$ such that $i + c m \leq n$, we have a shortest tour $\opt' \preceq \opt$ that is $2k$-tonic at all $s_i$ for $1 \leq i \leq n-cm$ and for $cm \leq i \leq n$.
Since $n \geq 2cm$, this implies that the tour $\opt'$ is $2k$-tonic at all $s_i$. \qedhere
\end{itemize}
\end{proof}
%------------------------------------------------------------------------
If the points of point set $P$ have distinct integer $x$-coordinates (note that we therefore have a sparse point set with $c = 2$), then we can get a slightly better bound in a similar way.
%------------------------------------------------------------------------
\begin{theorem}\label{thm:2ktonicproofs:main2k_int}
Let $P$ be a set of points with distinct integer $x$-coordinates inside
a $\delta$-strip.
Then there exists an optimal TSP tour on $P$ that is $2k$-tonic for $k \mydef \floori{ 2\sqrt{\delta + 1}}$.
\end{theorem}
%------------------------------------------------------------------------
\begin{proof}
The proof is analogous to that of Theorem~\ref{thm:2ktonicproofs:main2k_sparse}, substituting $\floori{k/2}$ for $cm$. The only significant difference is the proof that $(x_j - x_i) (k-j+i+1) \geq \delta$. This time, we have
$$(x_j - x_i) (k-j+i+1) \geq \floor{\frac{k}{2}} \br{k-\floor{\frac{k}{2}} + 1} = \floor{\frac{k}{2}} \br{\ceil{\frac{k}{2}} + 1}.$$
We finish the proof using that $k \in \mathbb{N}$ and $k \geq 2\sqrt{\delta + 1} - 1$:
$$\floor{\frac{k}{2}} \br{\ceil{\frac{k}{2}} + 1} \geq \frac{k-1}{2} \br{\frac{k+1}{2} + 1} \geq \frac{2\sqrt{\delta+1}-1-1}{2} \br{ \frac{2\sqrt{\delta+1}-1+1}{2} + 1} = \delta.$$
\end{proof}

Finally, we show that the factor $\sqrt{\delta}$ is necessary, by giving a set of examples where the shortest tour has exact tonicity $\Theta(\sqrt{\delta})$:

\begin{theorem}\label{thm:2ktonicproofs:stingray}
For any $k \geq 2$, there exists a sparse point set $P_k$ inside a $\delta$-strip of which the shortest tour is unique and has exact tonicity $2k$, with $\delta_k = \Theta (k^2)$.
\end{theorem}
\begin{proof}

Let $\delta_k \mydef 2k^2$.
First, we will give a set $Q_k$ of points with distinct integer $x$-coordinates of which the shortest \emph{path} is unique and consists of the shortest $|Q_k|-1$ edges.
Then, we will add a point $t$, and argue that the (unique) shortest tour on $P_k = Q_k \cup \{t\}$ is the aforementioned path with both endpoints connected to $t$.

Let $k \geq 2$ be any fixed integer. Let $p \mydef (2k+1,0)$, let $q_i \mydef \left(k+i,2k i \right)$ and let $r_i = \left(i,2k i+k\right)$.
Define
$$Q_k \mydef \{p\} \cup \left\{q_i \; \middle | \; i \in \{0, \ldots, k\} \right\} \cup \left\{r_i \; \middle | \; i \in \{0, \ldots, k-1\} \right\}.$$
See Figure~\ref{fig:2ktonicproofs:stingray} for an example with $k=2$.
\begin{figure}
\begin{center}
\includegraphics{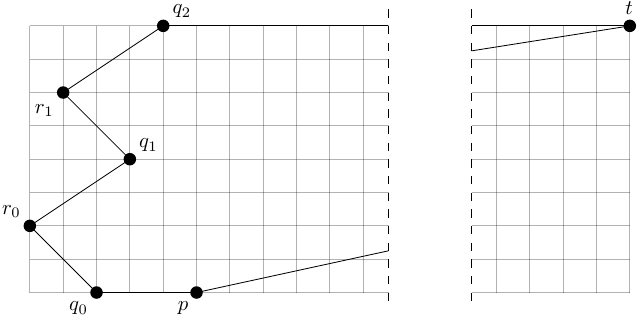}
\caption{A sketch of $Q_2'$ and its shortest tour, see Theorem~\ref{thm:2ktonicproofs:stingray}}
\label{fig:2ktonicproofs:stingray}
\end{center}
\end{figure}
We claim that $\pi \mydef (p,q_0,r_0,q_1,...,r_{k-1},q_k)$ is the unique shortest path on $Q_k$.
To see this, first observe that the edge $pq_0$ and all $k$ edges of the form $q_ir_i$ are the shortest possible edges, with lengths $k+1$ and length $\sqrt{2}k$, respectively.
The next shortest are the $k$ edges of the form $r_iq_{i+1}$, with length $\sqrt{2k^2+2k+1}$.
Note that these $2k+1$ edges together form the path $\pi$.
Therefore, if all other possible edges are indeed longer, then $\pi$ is indeed the shortest path through all points of $Q_k$, as claimed.
The only other serious candidates are of the form $r_ir_{i+1}$ or $q_iq_{i+1}$, all of which have length $\sqrt{4k^2+1}$.
Since $k \geq 2$, these are indeed longer than the edges in $\pi$.

Finally, let $t$ be any point in the $\delta$-strip with a sufficiently large $x$-coordinate.
Then, the two points closest to $t$ are those two that have the highest $x$-coordinates: $p$ and $q_k$.
Specifically, $t \mydef (3k^4, 2k^2)$ suffices.
Since $\pi$ is a shortest path on $Q_k$, and starts in $p$ and ends in $q_k$, we can combine it with the edges $q_kt$ and $tp$ to obtain the shortest tour $T$ through all points of $P_k = Q_k \cup \{t\}$.

In conclusion, for any $k \geq 2$, we can find a point set $P_k$, consisting of $3k^4+1$ points with distinct $x$-coordinates, and a $\delta_k$ of $2k^2$, of which the shortest tour is unique and has exact tonicity $2k$.
\end{proof}
%------------------------------------------------------------------------
%------------------------------------------------------------------------------------------

%------------------------------------------------------------------------------------------
%--------------------------------------------------------------------------------
\section{An algorithm for narrow cylinders}\label{sec:alg}
%--------------------------------------------------------------------------------
In this section we investigate how the complexity of \etsp depends on the
width~$\delta$ of the strip (or cylinder) containing the point set~$P$. Recall that a
point set~$P$ inside a $\delta$-cylinder is \emph{sparse} if for
every $x\in \Reals$ the set~$[x,x+1]\times \ball^{d-1}(\delta/2)$ contains $O(1)$~points.
%  Our main result is as follows.
%--------------------------------------------------------------------------------
\begin{theorem}\label{thm:alg}
Let $P$ be a set of $n$ points in a $\delta$-cylinder.
\begin{enumerate}
\item[(i)] If for any~$i\in \Integers$ the drum $[i\delta,(i+1)\delta]\times \ball^{d-1}(\delta/2)$
    contains at most $k$ points, then we can solve \etsp on~$P$ in $2^{O(k^{1-1/d})}n^2$ time.
\item[(ii)] If $P$ is sparse then we can solve \etsp in $2^{O(\delta^{1-1/d})}n^2$ time.
\end{enumerate}
\end{theorem}
%--------------------------------------------------------------------------------
\vspace*{2mm}
Part~(ii) of the theorem is a trivial consequence of part~(i), so the rest of the
section focuses on proving part~(i). Our proof uses and modifies some techniques
of~\cite{bbkk-ethtsp-2018}. For $i\in \Integers$, let $\sig_i$ be the drum
$[i\delta,(i+1)\delta]\times \ball^{d-1}(\delta/2)$. Define $n_i \mydef |\sig_i \cap P|$---we
assume without loss of generality that all points from $P$ lie in the interior of a drum~$\sig_i$---and let $k\mydef \max_i n_i$.  We say that a drum $\sig_i$ is \emph{empty} if $n_i=0$.
\medskip

We will regularly use that any subset~$E$ of edges from an optimal tour of a
point set~$P$ has the \emph{Packing Property}~\cite{bbkk-ethtsp-2018}:
for any $\lambda>0$ and any cube~$\sig$ of side length~$\lambda$, the number of edges from $E$ of length at least~$\lambda/4$ that intersect~$\sig$ is~$O(1)$.
The Packing Property is at the heart of several subexponential algorithms~\cite{Kann92,SmithW98}.
We need a recent separator theorem~\cite[Theorem 5]{bbkk-ethtsp-2018}, as explained next.

For a cube $\sig$ of side length $\lambda$ and a number $c>0$, let $c\sig$
denote the cube of side length~$c\lambda$ and with the same center as~$\sig$.
We say that an edge~$e$ \emph{enters} a cube~$\sig$ if one endpoint
of~$e$ is inside~$\sig$ and the other endpoint is outside~$\sig$.
The following lemma formalizes the core of the proof of Theorem~5 (and Corollary 3) 
in~\cite{bbkk-ethtsp-2018}. Theorem~5 from that paper is stronger than the lemma below, as it also involves balancing 
the number of points inside and outside the separator. To guarantee a good balance, 
a special cube $\sig$ is used in the lemma below, but for us that is
not relevant.
%-------------------------------------------------------------------------
\begin{lemma}[De Berg~\etal~\cite{bbkk-ethtsp-2018}]\label{thm:sep_ETSP_2}
Let $P$ be a unknown set of points in $\Reals^d$.
Let $Q$ be a (known) subset of~$P$ of $m$ points.
Let $T$ be an unknown shortest tour on $P$, and let $E$ be the (unknown) set of edges of $T$ of which both endpoints are in $Q$.
Then for any cube~$\sig$ we can in $O(m^{d+1})$ time compute a scale factor $c\in[1,3]$ such that the following holds for the cube~$c\sig$ (which is the cube $\sigma$ scaled by a factor $c$ with respect to its center).
\begin{itemize}
	\item $|E_{c\sig}|=O(m^{1-1/d})$, where $E_{c\sig}\subseteq E$ denotes
          the set of edges entering~$c\sig$.
	\item There is a family $\cC\subseteq 2^{E}$ of $2^{O(m^{1-1/d})}$ \emph{candidate sets}
          such that $E_{c\sig} \in \cC$, and this family can be computed in
          $2^{O(m^{1-1/d})}$ time.
\end{itemize}
\end{lemma}

Recall that a separator for a set $P$ of points inside a $\delta$-cylinder is a hyperplane orthogonal to the $x$-axis which does not contain a point from $P$ and which partitions $P$ into two non-empty subsets.
Let $T_{\myopt}$ be an optimal TSP tour on~$P$.
For a separator~$t$, let $T(t,\sig_i)$ denote the set of edges from $T_\myopt$ with both endpoints in $\sig_{i-1}\cup \sig_i\cup \sig_{i+1}$ and crossing~$t$.

%--------------------------------------------------------------------------------
\begin{lemma}\label{lem:random_det_sep}
Let $\sig_i$ be a drum as defined above.
Then we can compute a separator~$t$ intersecting~$\sig_i$ such that $|T(t,\sig_i)|=O(k^{1-1/d})$ in $O(k^{d+1})$ time. 
Furthermore, there is a family $\cC$ of $2^{O(k^{1-1/d})}$ sets, which we call \emph{candidate sets},
such that $T(t,\sig_i)\in \cC$, and this family can be computed in $2^{O(k^{1-1/d})}$ time.
\end{lemma}
%-------------------------------------------------------------------------
\begin{proof}
We apply Lemma~\ref{thm:sep_ETSP_2} to the point set
$Q \mydef P\cap (\sig_{i-1}\cup \sig_i\cup \sig_{i+1})$ and the
cube~$\sig$ of side length $\delta$ whose left facet contains
the right facet of~$\sig_{i}$.
The boundary $\bd (c\sig)$ of the cube $c\sig$ given by the lemma intersects $\sig_i$ (potentially on its boundary), and it is disjoint from the interiors of $\sig_{i-1}$ and $\sig_{i+1}$.
Let $\spt$ be the separator containing $\sig_i \cap \bd (c\sig)$. The number of points in the three drums is at most~$3k$.
The number of edges from $E$ crossing $\spt$ is $O((3k)^{1-1/d})=O(k^{1-1/d})$.
\end{proof}

%--------------------------------------------------------------------------------

%--------------------------------------------------------------------------------
\mypara{Separators and blocks.}
%--------------------------------------------------------------------------------
% Let $I=\{i_1,\dots,i_m\}$ be the set of indices with non-empty squares,
% that is, $n_i>0$ for some $i\in \Integers$ if and only if $i\in I$.
Consider the sequence of non-empty drums~$\sig_i$, ordered from left to right.
We use Lemma~\ref{lem:random_det_sep} to place a separator in every second drum
of this sequence.
Let $\SPT \mydef \{\spt_1,\ldots,\spt_{|\SPT|}\}$ be the resulting (ordered) set of separators.
Let $\spt_0$ and $\spt_{|\SPT|+1}$ denote separators coinciding with the left side of the leftmost non-empty drum and the right side of the rightmost non-empty drum, respectively;
see Fig.~\ref{fig:blocks} for an illustration in $\Reals^2$.
%--------------------------------------------------------------------------------
\begin{figure}
\begin{center}
\includegraphics{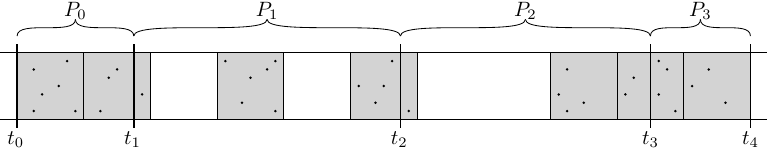}
\caption{The separators $\spt_0,\ldots,\spt_{|\SPT|+1}$ and the blocks they define.}
\label{fig:blocks}
\end{center}
\end{figure}
%--------------------------------------------------------------------------------
We call the region of the $\delta$-cylinder between two consecutive separators $\spt_i$
and $\spt_{i+1}$ a \emph{block}. Let $P_i\subseteq P$ denote the set of points in this block.
Note that $|P_i|\leq 3k$.

For an edge set $E$ and a separator~$\spt$, let $E(\spt)\subseteq E$ denote the subset of edges intersecting~$\spt$.
Define $P_{\myright}(E,\spt)$ to be the set of endpoints of the edges in $E(\spt)$ that lie to the right\footnote{The separators given by Lemma~\ref{lem:random_det_sep} are not incident to any input points.} of~$t$.
We call $P_{\myright}(E,\spt)$ the \emph{endpoint configuration of $E$ at $\spt$}.
The next two lemmas rule out endpoint configurations with
two ``distant'' points from the separator.

%--------------------------------------------------------------------------------
\begin{lemma}\label{lem:two_long_edges}
Let $s_{\myleft}: x=x_{\myleft}$ and $s_{\myright}: x=x_{\myright}$ be two separators
such that $x_{\myright}-x_{\myleft} > 3\delta$, and suppose there is a point $z\in P$
with $x_{\myleft} + 3 \delta / 2 < x(z) < x_{\myright} - 3 \delta/2$.
Then an optimal tour on $P$ cannot have two edges that both cross $s_{\myleft}$ and~$s_{\myright}$.
\end{lemma}
%--------------------------------------------------------------------------------
\begin{proof}
Suppose for a contradiction that an optimal tour~$T$ has two (directed) edges,
$q_1 q_2$ and $r_1 r_2$, that both cross $s_{\myleft}$ and~$s_{\myright}$. (The direction of $q_1 q_2$
and $r_1 r_2$ is according to a fixed traversal of the tour.) If both edges cross
$s_{\myleft}$ and $s_{\myright}$ from left to right (or both cross from right to left) then
replacing $q_1 q_2$ and $r_1 r_2$ by $q_1 r_1$ and $r_2 q_2$
% (and reversing the path from $q_2$ to $r_1$)
gives a shorter tour (see Observation~\ref{obs:2ktonicproofs:switcheroo}), leading to the desired contradiction.

Now suppose that $q_1 q_2$ and $r_1 r_2$ cross $s_{\myleft}$ and $s_{\myright}$ in opposite directions.
Assume w.l.o.g.~that $x(q_1) < x_{\myleft}$ and $x(r_2) < x_{\myleft}$, and that $z$ lies on the path
from $r_2$ to $q_1$. Let $u_1, \ldots, u_k$,  $v_1, \ldots, v_l$, $w_1, \ldots, w_m$ and $z_2$ be such that
\[
T = (q_1,q_2, u_1, \ldots, u_k, r_1,r_2, v_1, \ldots, v_l, z, z_2, w_1, \ldots, w_m, q_1).
\]
We claim that the tour $T'$ defined as
\[
T' = (q_1, r_2, v_1, \ldots, v_l, z, r_1, u_k, \ldots, u_1, q_2, z_2, w_1, \ldots, w_m, q_1)
\]
is a strictly shorter tour. To show this, we will first change our point set $P$ into
a point set~$P'$ such that if $\|T'\| < \|T\|$ on $P'$, then $\|T'\| < \|T\|$ also on $P$.
To this end we replace $q_1$ by $q'_1 \mydef q_1 q_2 \cap s_{\myleft}$ and
$q_2$ by $q'_2 \mydef q_1 q_2 \cap s_{\myright}$, and we replace $r_1$ by
$r'_1 \mydef r_1 r_2 \cap s_{\myleft}$ and $r_2$ by~$r'_2 \mydef r_1 r_2 \cap s_{\myright}$;
see Figure~\ref{fig:replaceP}.
%--------------------------------------------------------------------------------
\begin{figure}[b]
\begin{center}
\includegraphics{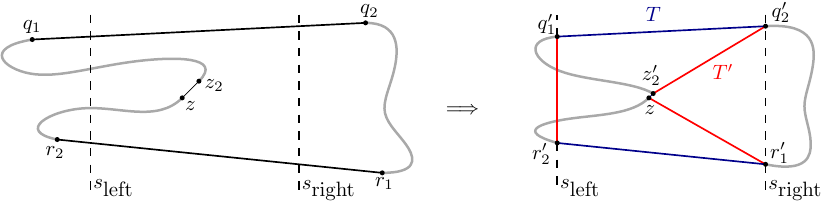}
\caption{Illustration for the proof of Lemma~\ref{lem:two_long_edges}.
The point $z'_2$ coincides with $z$ but is slightly displaced for visibility.
The sum of the lengths of the edges unique to $T$ (displayed in blue) is strictly larger than the sum of the lengths of the edges unique to $T'$ (displayed in red).}
\label{fig:replaceP}
\end{center}
\end{figure}
%--------------------------------------------------------------------------------

Finally, we replace $z_2$ by a point $z'_2$ coinciding with~$z$ (note that if $z_2 = q_1$, then we can split it before moving the resulting two points, analogous to the proof of Theorem \ref{thm:bitonic:main}).
Using a similar reasoning as in the proof of Lemma~\ref{lem:bitonic:reduction_to_consecutive},
one can argue that the point set~$P' := (P \setminus \{q_1,q_2,r_1,r_2,z_2\}) \cup \{q'_1,q'_2,r'_1,r'_2,z'_2\}$
has the required property.
To get the desired contradiction it thus suffices to show that $\|T\| - \|T'\| > 0$ on~$P'$.
This is true because
\[
\begin{array}{lll}
\|T\| - \|T'\| 	& = & |q_1' q_2'| + |r_1' r_2'| + |z z_2'| - |q_1' r_2'| - |z r_1'| - |q_2' z_2'| \\
			& \geq & |x_{\myright} - x_{\myleft}| + |x_{\myright} - x_{\myleft}| + 0\\
			& & \hspace{1cm}  - \delta - (|x_{\myright}- x(z)|+\delta) - (|x_{\myright} - x(z)|+\delta) \\
			& > & 2(x(z)-x_{\myleft}) - 3 \delta > 0,
\end{array}
\]
where the last line uses that $x_{\myleft} + 3 \delta / 2 < x(z)$.
\end{proof}
%--------------------------------------------------------------------------------

\begin{figure}[t]
\centering
\includegraphics{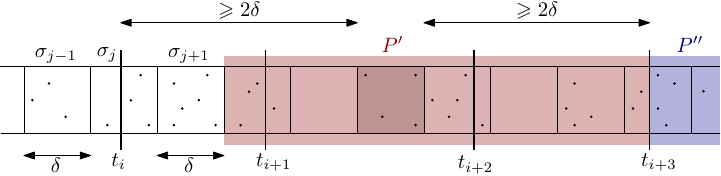}
\caption{Only $O(1)$ edges from $T_{\myopt}$ can cross $\spt_i$ and the hyperplane $x={(j+2)\delta}$, because of the Packing Property.
Tour edges crossing $\spt_i$ and $\spt_{i+3}$ obey Lemma~\ref{lem:two_long_edges}.
The points in the red area form $P'$. The points in the blue area form $P''$.
}\label{fig:twolongedges}
\end{figure}

%--------------------------------------------------------------------------------
\begin{lemma}\label{lem:sep_cross}
Let $\spt_i\in\SPT$ be a separator, and let $\sig_{j}=[j\delta,(j+1)\delta]\times\ball^{d-1}(\delta/2)$
denote the drum in which it is placed. Let $T_\myopt$ be an optimal tour on $P$ and let
$V := P_{\myright}(T_{\myopt},\spt_i)$ be its endpoint configuration at~$\spt_i$.
Let $P'$ denote the set of input points with $x$-coordinates between $(j+2)\delta$
and $x(\spt_{i+3})$, and let $P''$ be the set of input points with $x$-coordinate
larger than $x(\spt_{i+3})$.
Then (i) $|P' \cap V|\leq c^*$ for some absolute constant~$c^*$,
and (ii) $|P''\cap V|\leq 1$.
\end{lemma}
%--------------------------------------------------------------------------------
\begin{proof}
Consider a tour edge crossing $\spt_i$.
Any edge crossing $\spt_i$ ending in $P'$ must fully cross $\sig_{j+1}$, see Figure~\ref{fig:twolongedges}.
Therefore such edges have length at least~$\delta$.
By the Packing Property, there can be at most $c^* = O(1)$ such edges. This proves (i).

To prove (ii), note that there is a non-empty drum between $\spt_i$ and~$\spt_{i+3}$ with distance at least $2\delta$ from $\spt_i$ and $\spt_{i+3}$.
Lemma~\ref{lem:two_long_edges} thus implies that $T_\myopt$ has at most one edge crossing both $\spt_i$ and $\spt_{i+3}$, proving~(ii).
\end{proof}
%--------------------------------------------------------------------------------
Putting Lemma~\ref{lem:random_det_sep} and Lemma~\ref{lem:sep_cross} together,
we get the following corollary.
%--------------------------------------------------------------------------------
\begin{corollary}\label{cor:candidate-enumeration}
Let $T_{\myopt}$ be an optimal tour, let $\spt_i\in\SPT$ be a separator,
and let $V\subset P$ be the (unknown) endpoint configuration of $T_{\myopt}$ at $\spt_i$.
Then we can enumerate in $2^{O(k^{1-1/d})} n$ time a family $\cB_i$ of
candidate endpoint sets such that $V\in \cB_i$.
\end{corollary}
\begin{proof}
Let $c^*$, $P'$ and $P''$ be as defined in Lemma~\ref{lem:sep_cross}.
We use the enumeration of candidate edge sets from Lemma~\ref{lem:random_det_sep}, and only keep the endpoints that are to the right of $\spt_i$.
To each of these endpoint sets, we add at most $c^*$ endpoints from $P'$, and we add at most one endpoint from $P''$.
There are $\poly(k)$ and $O(n)$ ways to add such endpoints, respectively.
This gives a family of $2^{O(k^{1-1/d})}\cdot \poly(k)\cdot n=2^{O(k^{1-1/d})} n$ sets.
By Lemma~\ref{lem:sep_cross}, the resulting family of endpoint sets contains~$V$.
\end{proof}
%--------------------------------------------------------------------------------
In addition to the sets $\cB_i\ (i=1,\dots,|\SPT|)$, we define $\cB_0=\cB_{|\SPT|+1}\mydef \emptyset$.

%--------------------------------------------------------------------------------
\mypara{Matchings, the rank-based approach, and representative sets.}
%--------------------------------------------------------------------------------
When we cut a tour using a separator, the tour falls apart into
several paths. As in other TSP algorithms, we need to make sure that the paths
on each side of the separator can be patched up into a Hamiltonian cycle.
Following the terminology of~\cite{bbkk-ethtsp-2018}, let $P$ be our input point set,
and let $M$ be a perfect matching on a set~$B\subseteq P$, where the points of $B$
are called  \emph{boundary points}. A collection~$\cP=\{\pi_1,\ldots,\pi_{|B|/2}\}$
of paths \emph{realizes $M$ on $P$} if
(i) for each edge $(p,q)\in M$ there is a path $\pi_i\in \cP$ with~$p$ and~$q$ as endpoints,
and (ii) the paths together visit each point $p\in P$ exactly once.
We define the total length of $\cP$ as the length of the edges in its paths.
In general, the type of problem that needs to be solved on one side of a separator
is called \bdtsp.
The input to such a problem is a point set $P\subset \Reals^d$,
a set of boundary points $B\subseteq P$, and a perfect matching $M$ on $B$.
The task is to find a collection of paths of minimum total length that realizes~$M$ on~$P$.

To get the claimed running time, we need to avoid iterating over all matchings.
We can do this with the so-called \emph{rank-based approach}~\cite{single-exponential,CyganKN18}.
As our scenario is very similar to the general \etsp, we can reuse most definitions
and some proof ideas from~\cite{bbkk-ethtsp-2018}.

Let $\cM(B)$ denote the set of all perfect matchings on~$B$, and consider a
matching $M\in \cM(B)$. We can turn $M$ into a weighted matching by assigning
to it the minimum total length of any collection of paths realizing~$M$.
In other words, $\weight(M)$ is the length of the solution of \bdtsp for
input~$(P,B,M)$. We use $\cM(P,B)$ to denote the set of all such weighted
matchings on~$B$. Note that $|\cM(P,B)| = |\cM(B)| = 2^{O(|B|\log |B|)}$.

We say that two matchings $M,M' \in \cM(B)$ \emph{fit} if their union is a
Hamiltonian cycle in $B$. Consider a pair $P,B$. Let $\cR$ be a set of weighted matchings
on~$B$ and let $M$ be another matching on~$B$. We define
$\myopt(M,\cR) := \min \{ \weight(M') : M' \in \cR \text{ and }M' \mbox{ fits } M\}$, that is,
$\myopt(M,\cR)$ is the minimum total length of any collection of paths on $P$
that together with the matching $M$ forms a cycle.
A set $\cR \subseteq \cM(P,B)$ of weighted matchings  is defined to be
\emph{representative} of another set $\cR' \subseteq \cM(P,B)$
if for any matching $M\in \cM(B)$ we have $\myopt(M,\cR) = \myopt(M,\cR')$.
Note that our algorithm is not able to compute a representative set of $\cM(P,B)$,
because it is also restricted by the Packing Property and Lemma~\ref{lem:two_long_edges},
while a solution of \bdtsp for a generic $P,B,M$ may not satisfy them.
Let $\cM^*(P,B)$ denote the set of weighted matchings in $\cM(P,B)$ that have a
corresponding \bdtsp solution satisfying the Packing Property and Lemma~\ref{lem:two_long_edges}.

The basis of the rank-based approach is the following result.
%----------------------------------------------------------------------------
\begin{lemma} {\rm\bf [Bodlaender~\etal~\cite{single-exponential}, Theorem 3.7]} \label{lem:repr}
There exists a set $\overline{\cR}$ consisting of $2^{|B|-1}$ weighted matchings that is representative of the set $\cM(P,B)$.  Moreover, there is an algorithm \emph{Reduce} that, given
a representative set~$\cR$ of $\cM(P,B)$, computes such a set $\overline{\cR}$ in
$|\cR| \cdot 2^{O(|B|)}$ time.
\end{lemma}
%----------------------------------------------------------------------------
Lemma~\ref{lem:repr} can also be applied in our case, where $\cR$ is representative of $\cM^*(P,B)$, the set of weighted matchings in $\cM(P,B)$ that have a corresponding \bdtsp solution satisfying the Packing Property and Lemma~\ref{lem:two_long_edges}.

We say that perfect matchings $M$ on $B$ and $M'$ on $B'$ are
\emph{compatible} if their union on $B\cup B'$ is either a single cycle or a collection of paths. The \emph{join} of these matchings, denoted by $\mathrm{Join}(M,M')$ is a perfect matching on
the symmetric difference $B\symdiff B'$ obtained by iteratively contracting edges with an incident vertex of degree~$2$ in the graph $(B\cup B',M\cup M')$ .

\mypara{The algorithm.}
Our algorithm is a dynamic program, where we define a subproblem for each separator
index $i$, and each set of endpoints $B\in \cB_i$. The value of $A[i,B]$ is defined as follows.
\[A[i,B]:= \begin{cases}\parbox{0.8\textwidth}{A representative set containing pairs $(M,x)$, where $M$ is a perfect matching on~$B$
and $x$ is a real number equal to the minimum total length of a path cover of
$P_0\cup\dots\cup P_{i-1} \cup B$ realizing the matching $M$.}\end{cases}\]
 The length of the entire
tour will be the value corresponding to the empty matching at index $|\SPT|+1$, that is,
it will be the value~$x$ such that $A[|\SPT|+1,\emptyset]=\{(\emptyset,x)\}$.

Our dynamic-programming algorithm works on a block-by-block basis.
It uses the algorithm \emph{TSP-repr} by De Berg~\etal~\cite{bbkk-ethtsp-2018} for \bdtsp on arbitrary $d$-dimensional point sets to solve subproblems inside a block. In particular, \emph{TSP-repr} computes a set of weighted matchings that represents $\cM^*(P,B)$.
Algorithm~\ref{alg:sparse} gives our algorithm in a pseudocode, which is further explained below.
%----------------------------------------------------------------------------------
\begin{algorithm}[ht]
\caption{NarrowRectTSP-DP($P, \delta$)}\label{alg:sparse}%
\hspace*{\algorithmicindent} \textbf{Input:} A set $P$ of points in $(-\infty,\infty) \times \ball^{d-1}(\delta/2)$ that is sparse\\
\hspace*{\algorithmicindent} \textbf{Output:} The length of the shortest tour through all points in $P$
\begin{algorithmic}[1]
\State Compute the separators $\spt_0,\ldots,\spt_{|\SPT|+1}$ using Lemma~\ref{lem:random_det_sep}, as explained above. \label{step:separators}
\State Compute the sets $\cB_0, \ldots, \cB_{|\SPT|+1}$ as explained in the proof of Corollary~\ref{cor:candidate-enumeration}. \label{step:cbjs}
\State $A[0,\emptyset] \mydef \{(\emptyset,0)\}$
\For {$i = 1$ to $|\SPT|+1$}\label{step:outerloop}
	\ForAll {$B\in \cB_i$} \label{step:outer_bd_iter}
		\State $A[i,B] \mydef \emptyset$
		\ForAll {$B'\in \cB_{i-1}$ where $B' \subseteq P_{i-1} \cup B$}\label{step:innerboundary}
			\ForAll {$(M,x)\in \text{\emph{TSP-repr}}(P_{i-1} \cup B,\ B'\symdiff B)$}\label{step:blocksolve}
				\ForAll {$(M',x')\in A[i-1,B']$}	\label{step:prevA}		
					\If{$M'$ and $M$ are compatible}\label{step:checkcomp}
						\State Insert $(\mathrm{Join}(M,M'),x+x')$ into $A[i,B]$\label{step:insert}
					\EndIf
				\EndFor
			\EndFor
		\EndFor
		\State $Reduce(A[i,B])$\label{step:reduce}
	\EndFor
\EndFor
\State \Return $\length(A[|\SPT|+1,\emptyset])$
\end{algorithmic}
\end{algorithm}
%----------------------------------------------------------------------------------

The goal of Lines~\ref{step:outer_bd_iter}--\ref{step:reduce}
is to compute a representative set $A[i,B]$ of $\cM^*(P_0\cup\dots\cup P_{i-1} \cup B,\ B)$
of size $2^{O(k^{1-1/d})}$. We say that a point $p\in P$ is \emph{distant} (with respect to a
separator~$\spt_i$) if it is to the right of $\spt_{i+3}$, and denote the set of distant input points from $\spt_i$ by $\mathrm{distant}(\spt_i)$.
First, we iterate over all sets $B\in \cB_i$ in Line~\ref{step:outer_bd_iter}.
Next, we consider certain boundary sets $B'\in \cB_{i-1}$. Note that if there is a distant point $p\in B'\cap\,  \mathrm{distant}(\spt_i)$, then a tour edge crossing $\spt_{i-1}$ ending at $p$ also crosses $\spt_i$, and thus $p$ is also a (distant) point of $B$.

In Line~\ref{step:blocksolve} we call the algorithm of
De Berg~\etal~\cite{bbkk-ethtsp-2018} 
on the point set $P_{i-1}\cup B'\cup B$ with $B' \symdiff B$ as boundary set;
note that  $P_{i-1} \cup B' \cup B = P_{i-1} \cup B$ (since
$B'\subseteq P_{i-1}\cup B$), so the first parameter of the
call can be written as $P_{i-1}\cup B$.
This gives us a representative set
$\cR$ of $\cM^*(P_{i-1} \cup B,\ B'\symdiff B)$. For each weighted matching
$(M,x)\in \cR$, and for each weighted matching from the representative set
$(M',x')\in A[i-1,B']$, we check whether $M$ and $M'$ are compatible. If so,
then  taking the union of the corresponding path covers gives a path cover of
$P_0\cup\dots\cup P_{i-1} \cup B$ of total length $x+x'$, which realizes the matching
$\mathrm{Join}(M,M')$ on $B$; we then add $(\mathrm{Join}(M,M^*),x+x^*)$
to $A[i,B]$ in Line~\ref{step:insert}.

After iterating over all boundary sets $B'$, the entry $A[i,B]$ stores a set of weighted
matchings. These weighted matchings will be used on the next iteration at Line~\ref{step:prevA} to generate the $A[i+1,\cdot]$ entries. To prevent each iteration taking more time than the previous one, thus blowing up the constants, we reduce $A[i,B]$ to size at most $2^{M k^{1-1/d}}$ for some large but fixed $M$ using the
\emph{Reduce} algorithm~\cite{single-exponential} in Line~\ref{step:reduce}.

Note that in the final iteration, when $i=t+1$, we take $B=\emptyset$. Now
$M$ and $M'$ are compatible if and only if the union of the corresponding path covers
is a Hamiltonian cycle. Line~\ref{step:reduce} then gives to a single entry the smallest weight.
Therefore, the length of the only entry in $A[t+1,\emptyset]$ after the loops have ended,
is the length of the optimal TSP tour.
Hence, the correctness of NarrowRectTSP-DP follows from the next lemma.
%----------------------------------------------------------------------------------
\begin{lemma}~\label{lem:represent_this}
After Step~\ref{step:reduce}, the set $A[i,B]$ is a representative set of $\cM^*(P_0\cup\dots\cup P_{i-1} \cup B, B)$.
\end{lemma}

\begin{proof}
This proof is very similar to the correctness proof of~\cite{bbkk-ethtsp-2018}. We use induction on $i$. For $i=1$, we are directly calling \emph{TSP-repr}$(P_0 \cup B, B)$, which gives a representative set. This set is reduced to a potentially smaller set on Line~\ref{step:reduce}, which is still a representative set by Lemma~\ref{lem:repr}.

Assume now that $i>1$, and for each $B'\in \cB_{i-1}$ the family $A[i-1,B']$ is representative of $\cM^*(P_0\cup\dots\cup P_{i-2} \cup B', B')$. For any fixed $B'\in \cB_{i-1}$, we have that the set returned by \emph{TSP-repr} is a representative set of $\cM^*(P_{i-1} \cup B,\ B'\symdiff B)$. By Lemma~3.6 in~\cite{single-exponential}, the join operation executed on all compatible pairs of weighted matchings from two representative sets results in a representative set, thus for each fixed $B'$, we have inserted a representative set~of
\begin{align*}
\cM_{B'}\mydef \{\mathrm{Join}(M',M) \mid\ & M'\in \cM^*(P_0 \cup \dots\cup P_{i-2}\cup B',\ B'),\\
& M\in \cM^*(P_{i-1} \cup B,\ B'\symdiff B) \}
\end{align*}
into $A[i,B]$.
Consequently, just before executing Line~\ref{step:reduce} the set $A[i,B]$ is a representative set of $\bigcup_{B'\in \cB'} \cM_{B'}$, where $\cB'$ is the family of endpoint configurations we consider at Line~\ref{step:innerboundary}.
Note that $\cB'$ is a (super)set of all possible endpoint configurations in $\cB_{i-1}$ that a path cover of $P_0\cup\dots\cup P_{i-1} \cup B$ with boundary $B$ obeying the Packing Property and Lemma~\ref{lem:two_long_edges} can have.
The set $\cM_{B'}$ contains the subset of $\cM^*(P_0\cup\dots\cup P_{i-1},B)$ corresponding to path covers where the endpoint configuration at $\spt_{i-1}$ is $B'$.
Consequently, $\cM^*(P_0\cup\dots\cup P_{i-1},B)\subseteq A[i,B]$.
Therefore $A[i,B]$ is a representative set of $\cM^*(P_0\cup\dots\cup P_{i-1},B)$ just before executing Line~\ref{step:reduce}, and by Lemma~\ref{lem:repr}, it remains a representative set after executing this line.
\end{proof}
%----------------------------------------------------------------------------------

%----------------------------------------------------------------------------------
\mypara{Analysis of the running time.}
%----------------------------------------------------------------------------------
The loop of Lines~\ref{step:outerloop}--\ref{step:reduce} has $|\SPT|+1=O(n)$ iterations.
Each set $\cB_i$ contains $2^{O(k^{1-1/d})}n$ sets. For each choice of
$B\in \cB_i$, we have $2^{O(k^{1-1/d})}$ options for $B'$, since $B$ can have at most 
one point distant from $\spt_i$ by Lemma~\ref{lem:two_long_edges}.
The running time of \emph{TSP-repr} is $T(|P|,|B|)=2^{O(|P|^{1-1/d}+|B|)}$~\cite[Lemma~8]{bbkk-ethtsp-2018}.
By Lemma~\ref{lem:random_det_sep} we have $|B|=O(k^{1-1/d})$, so the running time of
each call to \emph{TSP-repr} in Algorithm~\ref{alg:sparse} is
\[2^{O((3k + |B|)^{1-1/d} + |B|^{1-1/d})}=2^{O(k^{1-1/d})}.\]
The representative set returned
by \emph{TSP-repr} has $2^{O(k^{1-1/d})}$ weighted matchings, and the representative set
of $A[i-1,B']$ also has $2^{O(k^{1-1/d})}$ matchings. Checking compatibility, joining and insertion in Lines~\ref{step:checkcomp} and~\ref{step:insert} takes $\poly(|M|,|M'|)=\poly(k)$ time.
Consequently, before executing the reduction in Line~\ref{step:reduce},
the set $A[i,B]$ contains at most $2^{O(k^{1-1/d})}\cdot 2^{O(k^{1-1/d})} = 2^{O(k^{1-1/d})}$ entries. 
By Lemma \ref{lem:repr} the application of the \emph{Reduce} algorithm results in a representative set of size at most $2^{|B|-1}=2^{O(k^{1-1/d})}$.
Hence, the total running time is
\[n\cdot 2^{O(k^{1-1/d})}n \cdot 2^{O(k^{1-1/d})} \cdot \left(2^{O(k^{1-1/d})} + 2^{O(k^{1-1/d})} poly(k)\right) = 2^{O(k^{1-1/d})}n^2.\]
This concludes the proof of Theorem~\ref{thm:alg}.
%------------------------------------------------------------------------------------------

%------------------------------------------------------------------------------------------
%----------------------------------------------------------------------------------
\section{Random point sets inside a narrow rectangle}\label{sec:random}
%----------------------------------------------------------------------------------
Algorithm \ref{alg:sparse} also works efficiently on random point sets, with only minimal changes.
Specifically, in this section we will prove the following.

\begin{theorem}\label{thm:algrandom}
Let $P$ be a set of $n$ points taken independently uniformly at random from the hypercylinder $[0,n] \times \ball^{d-1}(\delta/2)$.
Then we can solve \etsp on $P$ in $2^{O(\delta^{1-1/d})} n$ expected time.
\end{theorem}

Our proof has four steps.
In the first step, we will show that long edges are unlikely to be viable.
For the second step, recall the definition of the \emph{spacing} of $p_i$ (in $P$) as $\Delta_i = x_{i+1} - x_i$, for all $1 \leq i \leq n-1$.
We extend this definition with $\Delta_0 = x_1$, and $\Delta_n = n - x_n$.
Note that $\sum_{i=0}^n \Delta_i = n$.
The values $\Delta_i$ play an important role in the analysis of the algorithm, and it will be convenient to assume that the $\Delta_i$ are independent.
However, when the $x$-coordinates of the points are generated uniformly at random in the interval $[0,n]$, this is not quite true.
In the second step, we therefore describe a method to generate the random point set in a different way, and we show how to relate the expected running times in these two settings.
In the third step, we will explain which changes are made to the algorithm.
Finally, in the fourth step, we will use the previous steps to analyse the expected running time of the algorithm.

\mypara{Step 1: Long edges are unlikely to be viable.}
In this step, we will create a pattern of points through which no long edge can pass.
To do so, we will first show that if a long edge is in an optimal tour $T$, then $T$ must be bitonic at certain separators.
Recall that $\sps_j$ is the separator between points $p_j$ and $p_{j+1}$.
\begin{observation}\label{obs:random:long_edge_then_bitonic}
Let $i < j < k$ with $x_i + \delta^2 < x_j$ and $\Delta_j > 1$ and $x_{j+1} + \delta^2 < x_k$.
Let $T$ be an optimal tour on $P$, and suppose that $T$ contains $p_i p_k$.
Then $T$ is bitonic at $\sps_j$.
\end{observation}
\begin{proof}
See Figure~\ref{fig:random:long_edge_then_bitonic} for an example.
\begin{figure}
\begin{center}
\includegraphics{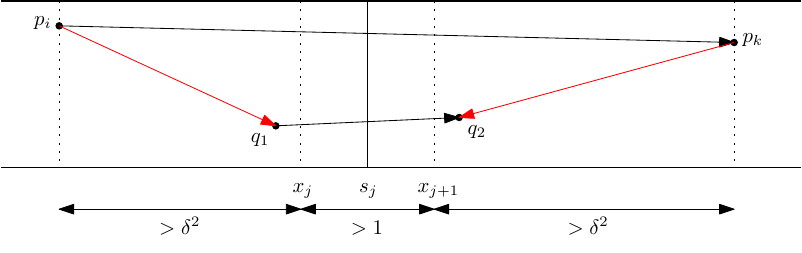}
\caption{An example of Observation~\ref{obs:random:long_edge_then_bitonic}.
Edges of $T$ are shown in black, edges of $T'$ are shown in red.}
\label{fig:random:long_edge_then_bitonic}
\end{center}
\end{figure}
W.l.o.g., $T$ contains the directed edge from $p_i$ to $p_k$.
Suppose for a contradiction that $T$ is an optimal tour and not bitonic at $\sps_j$.
Then $T$ must cross $s_j$ at least four times, which implies that there must be another directed edge $q_1 q_2$ with $x(q_1) \leq x_j$ and $x_{j+1} \leq x(q_2)$.
Note that $q_1 \neq p_i$ and $q_2 \neq p_k$, since a single point cannot have two incoming or two outgoing edges.
We will now show that $| p_i q_1| + |p_k q_2| < |q_1 q_2| + |p_i p_k|$.
By Observation~\ref{obs:2ktonicproofs:switcheroo}, we then get an optimal tour $T'$ shorter than $T$, thereby finishing the proof.
First, we will simplify the problem using an argument similar to the proof of Observation~\ref{obs:2ktonicproofs:linesum}.
Suppose we move $q_1$ towards $q_2$ until $x(q_1) = x_j$.
By doing so, $|q_1 p_i|$ will never decrease more than $|q_1 q_2|$.
Therefore, it is sufficient to prove the statement for $x(q_1) = x_j$.
Analogously, it is sufficient to prove the statement for $x(q_2) = x_j+1$ and $x_i = x_j - \delta^2$ and $x_k = x_j+1 + \delta^2$.
Together, we now have
$$x_i + \delta^2 = x(q_1) = x(q_2)-1 = x_k-1-\delta^2.$$
We get
\begin{align*}
|p_i q_1| + |p_k q_2|
    & \leq \sqrt{(x(q_1)-x_i)^2+\delta^2} + \sqrt{(x(q_2)-x_k)^2+\delta^2} \\
    & = \sqrt{(\delta^2)^2+\delta^2} + \sqrt{(\delta^2)^2+\delta^2} \\
    & = 2 \sqrt{ \delta^4 + \delta^2} \\
    & < 2 \sqrt{ \delta^4 + 2 \delta^2 + 1 } \\
    & = 2 \delta^2 + 2\\
    & = 1 + (2 \delta^2 + 1) \\
    & = (x(q_2)-x(q_1)) + (x_k-x_i) \\
    & \leq |q_1 q_2| + |p_i p_k|.
\end{align*}
By Observation~\ref{obs:2ktonicproofs:switcheroo}, this gives us an optimal tour $T'$ shorter than $T$, thereby finishing the proof.
\end{proof}
We are now ready to create our pattern.
We define a \emph{wall} to be a set of four points $(p_i, \ldots, p_{i+3})$ with the following properties:
\begin{itemize}
\item $\Delta_i > 7 \Delta_{i+1}$ and $\Delta_{i+1} > 1$ and $\Delta_{i+2} > 7 \Delta_{i+1}$, and
\item $|p_{i+1}' p_{i+2}'| > \delta / 2$,
\end{itemize}
where $p_{i+1}'$ and $p_{i+2}'$ are the orthogonal projections of $p_{i+1}$ and $p_{i+2}$ onto the ball $\{0\} \times \ball^{d-1}(\delta/2)$, respectively.
See Figure~\ref{fig:wall} for an example of a wall.
\begin{figure}
\begin{center}
\includegraphics{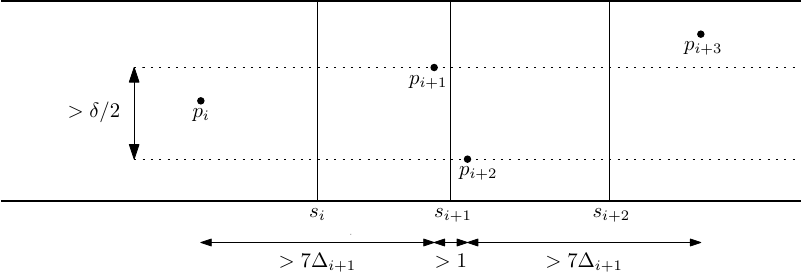}
\caption{An example of a wall, for $d = 2$.}
\label{fig:wall}
\end{center}
\end{figure}
Now we will prove that walls indeed tell us something about the validity of long edges.
To do so, we will need the following lemma:

\begin{lemma}\label{lem:random:three_sqrts}
Let $a, b, c > 0$.
If $4 a b - (4 c - a - b)^2 \geq 0$, then $\sqrt{a}+\sqrt{b} - 2 \sqrt{c} \geq 0$.
\end{lemma}
\begin{proof}
We split the proof into two cases:
\begin{itemize}
\item \textbf{Case $4c - a - b \geq 0$.} We get
\begin{align*}
\sqrt{a} + \sqrt{b} &= \sqrt{ a + b + \sqrt{4ab} } \\
&\geq \sqrt{ a + b + \sqrt{(4 c - a - b)^2} } & \text{since }4 a b - (4 c - a - b)^2 \geq 0\\
&= \sqrt{ a + b + 4c-a-b } & \text{since }4c - a - b \geq 0 \\
&= 2\sqrt{c}.
\end{align*}
\item \textbf{Case $4c - a - b < 0$.} We get
\begin{align*}
\sqrt{a} + \sqrt{b} &= \sqrt{ a + b + \sqrt{4ab} } \\
&> \sqrt{a + b} \\
&> 2\sqrt{c}. & \text{since }4c - a - b < 0
\end{align*}
\end{itemize}
This concludes the proof of Lemma~\ref{lem:random:three_sqrts}.
\end{proof}

We are now ready to prove that long edges passing through the wall must have one endpoint close to the wall.
\begin{lemma}\label{lem:random:walls}
Let $(p_i,\ldots,p_{i+3})$ be a wall.
Let $q_1q_2$ be an edge with $x(q_1) + \delta^2 < x_i$ and $x_{i+3} + \delta^2 < x(q_2)$.
Then the shortest tour on the points of $P$ does not contain $q_1q_2$.
\end{lemma}
\begin{proof}
Suppose for a contradiction that $T$ is an optimal tour that contains w.l.o.g. the directed edge $q_1q_2$.
By Observation~\ref{obs:random:long_edge_then_bitonic}, $T$ is bitonic at $s_i$ and $s_{i+1}$ and $s_{i+2}$.
Note that $q_1q_2$ passes through all three of these separators.
Therefore, one of the neighbours of $p_{i+1}$ must be to the left of $p_{i+1}$, and one must be to its right.
The same holds for $p_{i+2}$.
Since both $q_1q_2$ and the edge from $p_{i+2}$ to its `right' neighbour pass through $s_{i+2}$, the edge from $p_{i+1}$ to its `right' neighbour cannot pass through $s_{i+2}$.
Therefore, the `right' neighbour of $p_{i+1}$ is $p_{i+2}$.
In conclusion, $T$ must contain the directed edge $p_{i+2} p_{i+1}$.

It is sufficient to show that $|q_1 p_{i+2}| + |q_2 p_{i+1}| < |q_1 q_2| + |p_{i+2} p_{i+1}|$, because this will contradict the optimality of $T$ by Obervation~\ref{obs:2ktonicproofs:switcheroo}.
Analogously to the proof of Observation~\ref{obs:2ktonicproofs:linesum}, it is sufficient to prove this for $x(q_1) + \delta^2 + 7 \Delta_{i+1} = x_{i+1}$ and $x_{i+2} + \delta^2 + 7 \Delta_{i+1} = x(q_2)$.
Combining these simplifications, we want to prove that
\begin{align*}
&|q_1 q_2| + |p_{i+2} p_{i+1}| - |q_1 p_{i+2}| - |q_2 p_{i+1}|\\
&> \sqrt{(2\delta^2 + 15 \Delta_{i+1})^2+0^2} + \sqrt{\Delta_{i+1}^2 + \delta^2/4} - 2\sqrt{(\delta^2 + 8 \Delta_{i+1})^2 + \delta^2} \\
&\geq 0.
\end{align*}
We invoke Lemma~\ref{lem:random:three_sqrts} with $a=(2\delta^2 + 15 \Delta_{i+1})^2$ and $b= \Delta_{i+1}^2 + \delta^2/4$ and $c=(\delta^2 + 8 \Delta_{i+1})^2 + \delta^2$.
We check whether the requirement holds.
A bit of algebra gives us that for these values we can rewrite $4 a b - (4 c - a - b)^2$ to
$$\br{30 \Delta_{i+1} - \frac{225}{16}}\delta^4 + 4 \delta^6,$$
which is indeed at least $0$ for all $\delta > 0$ and $\Delta_{i+1} > 1$, as required.

Therefore, we have now proven that $|q_1 q_2| + |p_{i+2} p_{i+1}| - |q_1 p_{i+2}| + |q_2 p_{i+1}| > 0$.
Therefore, $T$ is not an optimal tour, leading to the conclusion that a shortest tour on the points of $P$ does not contain $q_1 q_2$.
\end{proof}

Note that if the $\Delta_i$ are i.i.d. (independent and identically distributed) and the other coordinates of the points are taken independently uniformly at random from $\ball^{d-1}(\delta/2)$, there exists some constant probability that $(p_i,\ldots,p_{i+3})$ is a wall.
This probability is independent of $\delta$, $n$, $i$, and whether any $(p_j,\ldots,p_{j+3})$ is a wall for any $j \leq i-3$ or $j \geq i+3$.
We will expand on this later.
This brings us to our next step.

\mypara{Step 2: Another way of generating points, with independent $\Delta_i$.}
As $n$ goes to infinity, the $\Delta_i$ become more and more independent.
True independence would greatly simplify the calculation of the running time.
Therefore, let us take a look at a different way of generating our point set.

First, we generate $n+1$ independent exponentially distributed variables, $\Delta_0,\ldots,\Delta_n \sim \Exp (1)$.
Using these, we compute the $x$-coordinates of our points, defined as $x_i := \sum_{j=0}^{i-1} \Delta_i$ for $1 \leq i \leq n$.
We will also write $x_{n+1} \mydef \sum_{i=0}^n \Delta_i$ for brevity.
Note that the value of $x_{n+1}$ will very likely be relatively close to $n$, since its distribution converges to a normal distribution with mean $n+1$ and standard deviation $\sqrt{n+1}$ when $n \rightarrow \infty$.
Finally, we generate the other coordinates by picking $n$ independent uniformly distributed points in $\ball^{d-1}(\delta/2)$.
Combining these gives us $n$ points in the hypercylinder $[0, x_{n+1}] \times \ball^{d-1}(\delta/2)$.

A point set generated this way has two important properties.
First of all, the $\Delta_i$ are now independent.
Second, as we will prove next, the expected running time of an algorithm on a uniformly distributed point set can be bounded by the expected running time of that algorithm on a point set generated this way.
We can even generalize this for more arbitrary point set distributions.

Let $X_n$ be a random point set of $n$ points in $\Reals^d$, where the $x$-coordinates of the points are taken independently uniformly at random from the interval $[0, n]$, and the other coordinates are i.i.d. and independent of the $x$-coordinates.
Let $Y_n$ be a random point set of $n$ points in $\Reals^d$, where the spacings $\Delta_i = x_{i+1}-x_i$ are taken independently from $\Exp(1)$ (including $\Delta_n$ to create an $x_{n+1}$), and the other coordinates are i.i.d. and independent of the $x$-coordinates, with the same distribution as in $X_n$.
We then get 
\begin{lemma}\label{lem:random:exptouni}
Let $X_n$ and $Y_n$ be as defined above.
Let $\A$ be an algorithm that runs on point sets in $\Reals^d$.
Let $T_\A(P)$ denote its running time for a point set $P$.
Suppose for all $\labda > 1$ and $P$, we have $T_\A(P_\labda) \leq T_\A(P)$, where $P_\labda$ is obtained from $P$ by scaling the $x$-coordinates by a factor $\labda > 1$.
Suppose $\EE[T_\A(Y_n)] = f(n)$ for some function $f$.
Then $\EE[T_\A(X_n)] = O(f(n))$.
\end{lemma} 
\begin{proof}
First, we note that $X_n$ and $Y_n$ are remarkably similar: when the $x$-coordinates of the points of $Y_n$ are scaled such that $x_{n+1} = n$, we obtain the same distribution as $X_n$, see \cite{daley2007introduction}.
Therefore, $$\EE[T_\A(X_n)] = \EE[T_\A(Y_n) | x_{n+1} = n].$$
Since $T_\A(P_\labda) \leq T_\A(P)$ for all $\labda > 1, P$, we get that $$\EE[T_\A(Y_n) | x_{n+1} = n] \leq \EE[T_\A(Y_n) | x_{n+1} \leq n].$$
Finally, since the distribution of $x_{n+1}$ converges to a normal distribution with mean $n+1$ and standard deviation $\sqrt{n+1}$ when $n \rightarrow \infty$, we have $\PP[x_{n+1} \leq n] \rightarrow \frac{1}{2}$ when $n \rightarrow \infty$. Therefore,
\begin{align*}
\EE[T_\A(Y_n) | x_{n+1} \leq n]
&< 3 \cdot \PP[x_{n+1} \leq n] \cdot \EE[T_\A(Y_n) | x_{n+1} \leq n] &\text{for $n$ large enough}\\
&\leq 3 \cdot \PP[x_{n+1} \leq n] \cdot \EE[T_\A(Y_n) | x_{n+1} \leq n] \\
&\; \; \; \; +3 \cdot \PP[x_{n+1} > n] \cdot \EE[T_\A(Y_n) | x_{n+1} > n] \\
&= 3 \cdot \EE[T_\A(Y_n)] \\
&= 3 \cdot f(n) = O(f(n)),
\end{align*}
finishing the proof.
\end{proof}

Now we can give a simple lower bound on the probability that $(p_i,\ldots,p_{i+3})$ forms a wall.
Let $B^{d-1}(p,r)$ be the $(d-1)$-dimensional ball with midpoint $p$ and radius $r$.
Define $q_1$ as $x(q_1) = 0, y_1(q_1) = 3 \delta /8, y_2(q_1) = \ldots = y_{d-2}(q_1) = 0$, and define $q_2$ as $x(q_2) = 0, y_1(q_2) = -3\delta/8, y_2(q_2) = \ldots = y_{d-2}(q_2) = 0$.
Recall that $p_{i+1}'$ and $p_{i+2}'$ are the orthogonal projections of $p_{i+1}$ and $p_{i+2}$, respectively, on $\{0\} \times \ball^{d-1}(\delta/2)$.
Now, note that if $p_{i+1}' \in \{0\} \times B^{d-1}(q_1,\delta/8)$ and $p_{i+2}' \in \{0\} \times B^{d-1}(q_2,\delta/8)$, then $|p_{i+1}'p_{i+2}'| \geq |q_1 q_2| - 2 \delta/8 = \delta/2$.
Then
\begin{align*}
& \PP[\Delta_i > 7 \Delta_{i+1} \wedge \Delta_{i+1} > 1 \wedge \Delta_{i+2} > 7 \Delta_{i+1} \wedge |p_{i+1}' p_{i+2}'| > \delta/2]\\
& > \PP[\Delta_i > 14 \wedge 1 < \Delta_{i+1} < 2 \wedge \Delta_{i+2} > 14 \wedge |p_{i+1}' p_{i+2}'| > \delta/2]\\
& = e^{-14} \cdot (e^{-1}-e^{-2}) \cdot e^{-14} \cdot \PP[ |p_{i+1}' p_{i+2}'| > \delta/2] & \text{since all $\Delta_j \sim \text{Exp}(1)$}\\
& = \Theta(1) \cdot \PP[ |p_{i+1}' p_{i+2}'| > \delta/2]\\
& > \Theta(1) \cdot \PP[ p_{i+1}' \in \{0\} \times B^{d-1}(q_1,\delta/8) ]^2 & \text{see Figure~\ref{fig:random:ballsinballs}}
\end{align*}
\begin{figure}
\begin{center}
\includegraphics{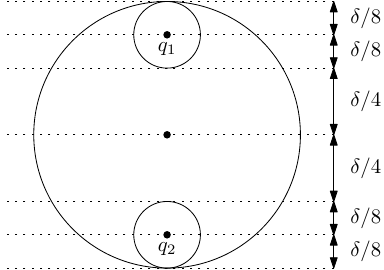}
\caption{An example with $d=3$.
The large circle is the separator $s_{i+1}$.
If $p_{i+1}'$ is in the upper circle $\{0\} \times B^{d-1}(q_1,\delta/8)$, and $p_{i+2}'$ is in the lower circle $\{0\} \times B^{d-1}(q_2,\delta/8)$, then $|p_{i+1}' p_{i+2}'|$ is at least $\delta/2$.
Note that the corresponding probabilities are equal, and that the location of $p_{i+1}'$ is independent of the location of $p_{i+2}'$.}
\label{fig:random:ballsinballs}
\end{center}
\end{figure}
Now, it is well-known that the volume of the $(d-1)$-dimensional ball with radius $r$ is $c_d r^{d-1}$ for some constant $c_d$.
Since $p_{i+1}'$ is a random point uniformly at random in $\{0\} \times \ball^{d-1}(\delta/2)$, we get
\begin{align*}
& \Theta(1) \cdot \PP[ p_{i+1}' \in \{0\} \times B^{d-1}(q_1,\delta/8) ]^2\\
& = \Theta(1) \cdot \br{ \frac{c_d (\delta/8)^{d-1} }{c_d (\delta/2)^{d-1} } }^2\\
& = \Theta(1) \cdot \br{ 2^{-2(d-1)} }^2 \\
& = \Theta(1),
\end{align*}
since we assume $d$ to be a constant.

\mypara{Step 3: A new algorithm.}
The new algorithm is exactly the same as the old algorithm, with three changes.
The first change is simple: we previously stated that we start by sorting our points on their $x$-coordinate.
In the case of sparse point sets we use a worst-case $O(n\log n)$ sorting algorithm for this, but now we can use a bucket sort with $n$ buckets.
This brings the expected sorting time to $O(n)$ \cite{clrs-ita-09}.

For the second change, we need to take another look at how we place the separators $t_i$.
We previously placed these separators in every second nonempty drum $\sigma_i \mydef [i\delta,(i+1)\delta]\times \ball^{d-1}(\delta/2)$ based on the points in $\sigma_{i-1} \cup \sigma_i \cup \sigma_{i+1}$.
However, in order for our algorithm to meet the requirements of Lemma~\ref{lem:random:exptouni}, we would like to avoid having a point enter a drum after the $x$-coordinates are multiplied by some factor $\labda > 1$.
Furthermore, since the proof of Lemma~\ref{lem:random_det_sep} requires every drum to be at least $\delta$ wide, we cannot simply scale the drums as well.
Therefore, we will need to allow some overlap. 
For $j \in \{-1,0,1\}$ and all $i$, we define $\tau_{i,j} \mydef [(X(i)-1/2 + j)\delta,(X(i)+1/2+j)\delta] \times \ball^{d-1}(\delta/2)$, with $X(i) \mydef (i+1/2)x_{n+1}/n$.
Note that the center of drum $\tau_{i,0}$ scales with the point set, and $\tau_{i,-1}$ and $\tau_{i,1}$ are the drums directly to its left and right.
In our new algorithm, we will use the drums $\tau_{i,j}$ instead of the $\sigma_i$ used so far.
Since this is quite a drastic change, before we show how exactly these new drums are used in the algorithm, we will need to revisit all drum-related lemmas and corollaries.
Let us start with Lemma~\ref{lem:random_det_sep}.
When applied to the new drums, it gives us the following.

\begin{lemma}\label{lem:random_det_sep_2}
Let $\opt$ be an optimal TSP tour on $P$.
For a separator $t$, let $T(t,\tau_{i,0})$ denote the set of edges from $\opt$ with both endpoints in $\tau_{i,-1} \cup \tau_{i,0} \cup \tau_{i,1}$ and crossing $t$.
Let $k$ be the number of points in $\tau_{i,-1} \cup \tau_{i,0} \cup \tau_{i,1}$.
Then in $O(k^{d+1})$ time we can compute a separator~$t$ intersecting~$\tau_{i,0}$ such that $|T(t,\tau_{i,0})|=O(k^{1-1/d})$ . 
Furthermore, there is a family $\cC$ of $2^{O(k^{1-1/d})}$ candidate sets such that $T(t,\tau_{i,0})\in \cC$, and this family can be computed in $2^{O(k^{1-1/d})}$ time.
\end{lemma}
\begin{proof}
The proof is analogous to that of Lemma~\ref{lem:random_det_sep}.
\end{proof}
Note that $k$ is now defined `locally', since we no longer have a good bound on the maximum density of our point set.
Instead of placing the separators $t_i$ in every second nonempty $\sigma_i$, we place them in every $\tau_{i,0}$ instead, using the points in $\tau_{i,-1} \cup \tau_{i,0} \cup \tau_{i,1}$.
We continue to Lemma~\ref{lem:sep_cross}, since Lemma~\ref{lem:two_long_edges} is independent of the drums used.
We first define $P[i,j] \mydef \{p_i, \ldots, p_j\}$ for all $1 \leq i \leq j \leq n$.
The revised version of Lemma~\ref{lem:sep_cross} is the following:
\begin{lemma}\label{lem:sep_cross_2}
Let $\spt_i\in\SPT$ be a separator.
Let $\opt$ be an optimal tour on $P$ and let $V := P_{\myright}(\opt,\spt_i)$ be its endpoint configuration at~$\spt_i$.
Let $z$ be the leftmost point with an $x$-coordinate larger than $(X(i)+3)\delta$.
Let $p_j$ be the leftmost point more than $3\delta/2$ to the right of $z$.
Let $P'$ denote the set of points with $x$-coordinates between $(X(i)+3/2)\delta$ (the right border of $\tau_{i,1}$) and $x_j$, and let $P'' = P[j,n]$.
Then (i) $|P' \cap V|\leq c^*$ for some absolute constant~$c^*$,
and (ii) $|P''\cap V|\leq 1$.
\end{lemma}
\begin{proof}
Consider a tour edge crossing $\spt_i$.
Any edge crossing $\spt_i$ ending in $P'$ must fully cross $\tau_{i,1}$.
Therefore such edges have length at least~$\delta$.
By the Packing Property, there can be at most $c^* = O(1)$ such edges. This proves (i).

To prove (ii), we can directly apply Lemma~\ref{lem:two_long_edges}. Note that point $z$ indeed suffices, since it is more than $3\delta/2$ to the right of the right border of $\tau_{i,1}$.
\end{proof}

Next up is Corollary~\ref{cor:candidate-enumeration}, which we adapt as follows:
\begin{corollary}\label{cor:candidate-enumeration_2}
Let $\opt$ be an optimal tour.
Let $\spt_i\in\SPT$ be a separator.
Let $V\subset P$ be the (unknown) endpoint configuration of $\opt$ at $\spt_i$.
Then in $2^{O(k^{1-1/d})} n$ time we can enumerate a family $\cB_i$ of candidate endpoint sets such that $V \in \cB_i$.
Here, $k$ denotes the number of points which either lie in $\tau_{i,-1} \cup \tau_{i,0} \cup \tau_{i,1}$ or are an element of $P'$, where $P'$ is as defined in Lemma~\ref{lem:sep_cross_2}.
\end{corollary}
\begin{proof}
Let $c^*$, $P'$ and $P''$ be as defined in Lemma~\ref{lem:sep_cross_2}.
We use the enumeration of candidate edge sets from Lemma~\ref{lem:random_det_sep_2}, and only keep the endpoints that are to the right of $\spt_i$.
To each of these endpoint sets, we add at most $c^*$ endpoints from $P'$, and we add at most one endpoint from $P''$.
There are $\poly(k)$ and $O(n)$ ways to add such endpoints, respectively.
This gives a family of $2^{O(k^{1-1/d})} \cdot \poly(k) \cdot n = 2^{O(k^{1-1/d})} n$ sets.
By Lemma~\ref{lem:sep_cross_2}, the resulting family of endpoint sets contains~$V$.
\end{proof}

The rank-based approach is independent of the drums used.
This brings us to the algorithm itself.
It requires only a single change: all separators to the left of a previously found separator can be ignored.

For the third and final change, we can use walls to decrease the number of viable distant points.
If we find a wall sufficiently far to the right of $\spt_i$, then by Lemma~\ref{lem:random:walls} all points that are sufficiently far to the right to this wall are unviable and can therefore be left out of $P''$.
Specifically, let $p_j$ be the leftmost point with $x_j > X(i)+\delta/2+\delta^2$ and $j$ divisible by $4$.
Note that this point is at least $\delta^2$ to the right of $\tau_{i,0}$ and therefore at least $\delta^2$ to the right of $\spt_i$.
We then check whether $(p_j,\ldots,p_{j+3})$ is a wall.
If it is not, we check $(p_{j+4},\ldots,p_{j+7})$ and so forth, until we either have found a wall, or have reached $p_n$.
If we find a wall, then we only use the distant points that are not rendered unviable by the wall.
If we do not, then we use the same set of distant points as we used previously.

\mypara{Step 4: Analysis of the running time.}
As stated before, the sorting of the points can be done in $O(n)$ expected time.
For the rest of the analysis, we will need a few lemmas.
\begin{observation} \label{obs:random:probbound}
Let $X,Y$ be two integer nonnegative random variables, such that for all $k \geq 0$, the equation $\PP[X \leq k] \geq \PP[Y \leq k]$ holds.
Let $f(k)$ be an increasing nonnegative function such that $\EE[f(Y)] < \infty$.
Then
$$\EE[f(X)] \leq \EE[f(Y)].$$
\end{observation}
For the proof, see~\cite[Formula 3.3 on p. 6]{thorisson}.

The next lemma gives us information about among others the running time of an algorithm when the input has a random size.
\begin{lemma}\label{lem:random:sqrt_algo_on_poisson}
Let $f(x)$ be a function such that $f(x) = 2^{O(x^{1-1/d})}$ for some $d \geq 2$.
Let $X$ be a Poisson distributed random variable with mean $z$.
Then $\EE[f(X)] = 2^{O(z^{1-1/d})}$.
\end{lemma}
\begin{proof}
Let $M, \labda > 0$ be such that for all $x > 0$ we have $f(x) \leq M + 2^{\labda x^{1-1/d}}$.
Let us write $\mu \mydef e^{\labda+2}$.
We get
\begin{align*}
\EE[f(X)] - M &\leq \EE[2^{\labda X^{1-1/d}}] \\
&\leq 2^{\labda \floor{\mu z}^{1-1/d}} \PP[X \leq \floor{\mu z}] +  \sum_{k=\ceili{\mu z}}^{\infty} 2^{\labda k^{1-1/d}} \PP[X = k] \\
&= 2^{\labda \floor{\mu z}^{1-1/d}} \PP[X \leq \floor{\mu z}] +  \sum_{k=\ceil{\mu z}}^{\infty} 2^{\labda k^{1-1/d}} \frac{z^k e^{-z}}{k!}\\
&< 2^{\labda \mu z^{1-1/d}} +  \sum_{k=\ceil{\mu z}}^{\infty} 2^{\labda k^{1-1/d}} \frac{z^k e^{-z}}{(k/e)^k}\\
&< 2^{\labda \mu z^{1-1/d}} +  \sum_{k=\ceil{\mu z}}^{\infty} e^{\labda k} \frac{z^k e^k}{(\mu z)^k}\\
&= 2^{\labda \mu z^{1-1/d}} +  \sum_{k=\ceil{\mu z}}^{\infty} e^{-k} = 2^{\labda \mu z^{1-1/d}} + O(1).
\end{align*}
We have now shown that $\EE[f(X)] \leq 2^{\labda \mu z^{1-1/d}} + M + O(1)$, finishing the proof.
\end{proof}

This brings us to the result we need:
\begin{lemma}\label{lem:random:poisson_magic}
Let $X$ be a Poisson distributed random variable with mean $z$.
Let $i \geq 0$ be an arbitrary but fixed integer. 
Let $j$ be any integer with $i \leq j$.
Let $M$ be a random variable defined by $\PP[M = k] \mydef \PP[X = k] / \PP[X \in \{i,\ldots,j\} ]$ for $k \in \{i,\ldots,j\}$, and $\PP[M = k] = 0$ otherwise.
Let $f(x)$ be such that $f(x) = 2^{O(x^{1-1/d})}$.
Then $\EE[f(M)] = 2^{O(z^{1-1/d})}$.
\end{lemma}

\begin{proof}
Let $\labda, \mu > 0$ be such that $g(x) \mydef \mu 2^{\lambda x^{1-1/d}}$ is larger than $f(x)$ for all $x$.
Clearly, it is sufficient to show that $\EE[g(M)] = 2^{O(z^{1-1/d})}$.

Let $X' \mydef X+i$.
First, we will show that $\PP[M \leq k] \geq \PP[X' \leq k]$ for all $k$:
\begin{itemize}
\item Suppose $0 \leq k \leq i-1$.
Since $\PP[M \leq k] = \PP[X' \leq k] = 0$, the statement indeed holds.
\item Suppose $k \geq j$.
Since $\PP[M \leq k] = 1$ and $\PP[X' \leq k] < 1$, the statement indeed holds.
\item Suppose $i \leq k \leq j-1$.
For this case, let us first take a look at the value of $\PP[M = m] - \PP[X' = m]$:
\begin{align*} 
\PP[M = m] - \PP[X' = m]
&= \frac{e^{-z} z^m}{m! \PP[i \leq X \leq j]} - \frac{e^{-z} z^{m-i}}{(m-i)!} \\
&= \frac{e^{-z} z^{m-i}}{m!} \br{\frac{z^i}{\PP[i \leq X \leq j]} - \frac{m!}{(m-i)!}}
\end{align*}
Note that $\PP[M = m] - \PP[X' = m] < 0$ if and only if $m$ is large enough.
Therefore, $\PP[M \leq k] - \PP[X' \leq k]$ is minimal at $k = i-1$ or $k = j$.
We have already seen that for both of these values of $k$, we have $\PP[M \leq k] - \PP[X' \leq k] \geq 0$.
Therefore, for all $i \leq k \leq j-1$, we have $\PP[M \leq k] \geq \PP[X' \leq k]$.
\end{itemize}
Since $\PP[M \leq k] \geq \PP[X' \leq k]$ for all $k$, and $g$ is increasing, Observation~\ref{obs:random:probbound} gives us that $\EE[g(M)] \leq \EE[g(X')]$.
Furthermore, since
$$g(x+i) = \mu 2^{\lambda (x+i)^{1-1/d}} \leq \mu 2^{\lambda (x^{1-1/d} + i)} = 2^{\lambda i} g(x),$$
we get
$$\EE[g(X')] = \EE[g(X+i)] \leq 2^{\lambda i} \EE[g(X)] = O(\EE[g(X)]).$$
Since $X$ follows a Poisson distribution with mean $z$, Lemma~\ref{lem:random:sqrt_algo_on_poisson} gives us that $\EE[g(X)] = 2^{O(z^{1-1/d})}$.
In conclusion, we have
$$\EE[f(M)] \leq \EE[g(M)] \leq \EE[g(X')] = O(\EE[g(X)]) = 2^{O( z^{1-1/d})},$$
as we wanted to prove.
\end{proof}

We are now ready to analyse the total running time of the algorithm.
We will look at each part of the algorithm separately.
\begin{itemize}
\item \textbf{Line~\ref{step:separators}: Computing the separators.}
There are $O(n)$ separators to be computed.
Recall that $\tau_{i,j} = [(X(i)-1/2 + j)\delta,(X(i)+1/2+j)\delta] \times \ball^{d-1}(\delta/2)$, with $X(i) \mydef (i+1/2)x_n/n$ for $j \in \{-1,0,1\}$.
By Lemma~\ref{lem:random_det_sep_2}, each separator is computed in $O(M^{d+1})$ time, where $M$ is the number of points in $\tau_{i,-1} \cup \tau_{i,0} \cup \tau_{i,1}$.
Now we want a bound on the total expected amount of time needed per separator, $\sum_{k=0}^\infty k^{d+1} \PP[M = k]$.
Recall that $\Delta_j \sim \Exp (1)$ for all $j$, and the total horizontal interval covered by $\tau_{i,-1} \cup \tau_{i,0} \cup \tau_{i,1}$ is $3 \delta$.
Clearly, $M \geq 0$.
Therefore, $M$ is almost a Poisson distributed variable with expected value $3\delta$. There is one difference: since there are only $n$ points, $M \leq n$.
Therefore, $M$ does have all the required properties for Lemma~\ref{lem:random:poisson_magic}, so the total expected time needed per separator is $2^{O((3 \delta)^{1-1/d})} = 2^{O(\delta^{1-1/d})}$.
(Using \cite{riordan1937}, we can prove that the expected time is $O(\delta^{d+1})$, but the above will suffice.)
Since there are $O(n)$ separators, the total expected time needed for this part is $2^{O(\delta^{1-1/d})} n$.
\item \textbf{Lines~\ref{step:cbjs} and~\ref{step:outer_bd_iter}: The sets $\cB_i$}.
We will now give an upper bound on the expected value of $\sum_{i=1}^{|\SPT|+1} |\cB_i|$.
Recall that each $B \in \cB_i$ consists of (i) the rightmost endpoints of a candidate edge set, (ii) at most $c^*$ points from $P'$ and (iii) at most one point to the right of $P'$, not rendered unviable by a wall.
Note that the expected number of possibilities of each type is dependent: the more possibilities one set contains, the fewer others can have.
After all, no point can be in two sets of different types at once.
Aside from the number of points in each interval, the three types are otherwise independent.
Therefore, to find an upper bound on $\EE[|\cB_i|]$ it suffices to multiply upper bounds of the expected values of possibilities of (i), (ii) and (iii).
Let us start with the number of viable distant points, (iii).
Let $p_j$ be the leftmost point such that $(p_j,\ldots,p_{j+3})$ is the first candidate wall checked.
We have already seen that the probability that it is a wall is some constant $\alpha$.
If it is not a wall, the next set we check has the same probability, and so forth.
Therefore, the geometric distribution with probability $\alpha$ forms an upper bound on the expected number of possible walls checked.
The expected $x$-coordinate of the last viable distant point therefore is $X(i) + O(\delta^2) + O(1/\alpha) + O(\delta^2) = x(t_i) + O(\delta^2)$.
Therefore, the expected number of viable distant points is $O(\delta^2)$.

This brings us to the number of possible $B'$.
Since we use Corollary~\ref{cor:candidate-enumeration_2}, this is $\delta^2 \cdot 2^{O(M^{1-1/d})}$, where $M$ is the number of points which are in $\tau_{i,-1} \cup \tau_{i,0} \cup \tau_{i,1}$ or in $P'$.
Note that we have a factor $O(\delta^2)$ instead of the factor $O(n)$ as in the statement of Corollary~\ref{cor:candidate-enumeration_2}. This is because, as we have just shown, the expected number of distant points we need to consider is not $O(n)$ but $O(\delta^2)$.
Now, let us look at $M$.
Recall how $P'$ was defined: 
Let $z$ be the leftmost point with an $x$-coordinate larger than $(X(i)+3)\delta$.
Let $p_j$ be the leftmost point more than $3\delta/2$ to the right of $z$.
Then $P'$ contains all points between $(X(i)+3/2)\delta$ (the right border of $\tau_{i,1}$) and $p_j$.
Therefore, $M$ equals the number of points in an interval of length $3\delta+3\delta/2+3\delta/2 = 6\delta$, plus one.
We once more apply Lemma~\ref{lem:random:poisson_magic}, giving us that the total expected number of $B'$ is $2^{O(\delta^{1-1/d})}$.

Since $i \in \{1,\ldots,|\SPT|+1\} \subseteq \{1,\ldots,n\}$, the total expected number of elements in all $\cB_i$ combined is bounded by $\delta^2 \cdot 2^{O(\delta^{1-1/d})} n = 2^{O(\delta^{1-1/d})} n$.
\item \textbf{Lines~\ref{step:innerboundary}-\ref{step:reduce}}.
The running time of \emph{TSP-repr} is $T(|P|,|B|)=2^{O(|P|^{1-1/d}+|B|)}$~\cite[Lemma~8]{bbkk-ethtsp-2018}.
By Lemma~\ref{lem:random_det_sep_2} we have $|B|=O(M_1^{1-1/d})$, so the running time of each call to \emph{TSP-repr} in Algorithm~\ref{alg:sparse} is $2^{O(M_2^{1-1/d})}$, for some $M_1, M_2$ with all properties required for Lemma~\ref{lem:random:poisson_magic}.
Similar statements hold for the number of weighted matchings in the representative set returned by \emph{TSP-repr} and the number of matchings in the representative set of $A[i-1,B']$.
Consequently, before executing the reduction in Line~\ref{step:reduce},
the set $A[i,B]$ contains at most $2^{O(M_1^{1-1/d})}\cdot 2^{O(M_2^{1-1/d})} = 2^{O(M_3^{1-1/d})}$ entries for some $M_1, M_2, M_3$ with all properties required for Lemma~\ref{lem:random:poisson_magic}.
By Lemma \ref{lem:repr} the application of the \emph{Reduce} algorithm results in a representative set of size at most $2^{|B|-1}=2^{O(M^{1-1/d})}$ for some $M$ with all properties required for Lemma~\ref{lem:random:poisson_magic}.
\end{itemize}
In conclusion, the total expected running time is bounded by $2^{O(\delta^{1-1/d})}n$ when the $\Delta_i$ are exponentially distributed.

Finally, we will show that the requirements for Lemma~\ref{lem:random:exptouni} hold, where we take $\A$ to be the algorithm described above.
The only nontrivial requirement is that $T_\A(P_\labda) \leq T_\A(P)$ for all point sets $P$ and $x$-axis scaling factors $\labda > 1$.
Intuitively, we would expect this to hold: every single separator $t_i$ and wall $(p_i,\ldots,p_{i+3})$ retains its properties when scaled with a factor larger than $1$, while the boundary sets can only decrease in size and $P_\labda$ might contain extra walls.
Furthermore, since we only check whether $(p_i,\ldots,p_{i+3})$ forms a wall for $i$ divisible by $4$, we do not miss walls we previously would have seen.

Unfortunately, $T_\A(P_\labda) \leq T_\A(P)$ does not necessarily hold for all $\labda$ and $P$, since we have only proven bounds on the running times, not the actual running times themselves.
However, it is easy to show the existence of an algorithm $\B$ of which we know the exact expected running time of every step.
Algorithm $\B$ is simply algorithm $\A$, but after every step it waits as long as necessary to make its expected running time for that step equal to the bound calculated for this step.
To be precise, there are two types of waiting, best explained by an example.
We have previously shown that computing a separator takes $O(|Q|^{d+1})$ time, where $Q$ is the set of points of $P$ in $\tau_{i,-1} \cup \tau_{i,0} \cup \tau_{i,1}$.
Let $c_1$ be such that for all $Q$, computing this separator takes $T(Q) < c_1 |Q|^{d+1}$ time for $\A$.
When algorithm $\B$ computes a separator, it waits for $c_1 |Q|^{d+1} - T(Q)$ time afterwards, therefore taking exactly $c_1 |Q|^{d+1}$ time in total.
Note that the time \emph{waited} by $\B$ depends on $Q$, but the \emph{total} time used depends only on $|Q|$, and not $Q$ itself.
This is the end of the first type of waiting, but not of the example.
Next, we showed that the total \emph{expected} time needed to compute a separator was upper bounded by $2^{O(\delta^{1-1/d})}$.
So, let $c_2$ be such that $\A$, including the first type of waiting, takes $T(\delta,d) < 2^{c_2 \delta^{1-1/d}}$ expected time to compute a separator.
When computing a separator, $\B$ also waits $2^{c_2 \delta^{1-1/d}} - T(\delta,d)$ time.
Note that the time waited is independent of $Q$.
Together, these two types of waiting ensure that (i) the time needed by $\B$ is monotone in $|Q|$ and (ii) the total expected time needed by $\B$ equals the calculated upper bound for $\A$.
Finally, we always wait as long as needed to emulate the separators being in the worst possible place for that specific step.
Therefore, the time taken is completely independent of the locations of the separators.

Since we know the exact running time of $\B$, using the arguments given earlier we get that $T_\B(P_\labda) \leq T_\B(P)$ for all point sets $P$ and $x$-axis scaling factors $\labda > 1$.
Therefore, we can apply Lemma~\ref{lem:random:exptouni} on $\B$, giving us that $\B$ has an expected running time of $2^{O(\delta^{1-1/d})} n$ on $X_n$, the random point set of which the $x$-coordinates are picked uniformly from $[0,n]$.
Since $\B$ is slower than $\A$, we conclude that the expected running time of $\A$ must have the same bound, finishing the proof.
%------------------------------------------------------------------------------------------

%------------------------------------------------------------------------------------------
\section{A more efficient algorithm for sparse point sets}\label{sec:sparsev2}
Recall that a point set~$P$ inside a $\delta$-cylinder is \emph{sparse} if for every $x\in \Reals$ the set~$[x,x+1]\times \ball^{d-1}(\delta/2)$ contains at most $O(1)$~points.
We have already proven that if $P$ is sparse, we can solve \etsp in $2^{O(\delta^{1-1/d})} n^2$ time;
see Theorem~\ref{thm:alg}.
The goal of this section is to obtain the following improved result:

\begin{theorem}\label{thm:sparsev2:runtime}
Let $P$ be a set of $n$ points in a $\delta$-cylinder.
If $P$ is sparse then we can solve \etsp in $2^{O(\delta^{1-1/d})} n + O(\delta^2 n^2)$ time.
\end{theorem}

Our algorithm will be a dynamic program which uses three different tables.
The first two of these will be very similar to the table used by our previous sparse point set algorithm.
We will proceed in four steps.
First, we will recall a few core concepts, and introduce another set of separators.
Second, we will introduce the three tables, and define what values they contain.
Third, we will show how to compute these values.
Fourth, we will show that the algorithm runs in the desired time.

\mypara{Step 1: A new set of separators.}
Because $P$ is sparse, the drum $\sigma_i \mydef [(i-1)\delta,i\delta] \times \ball^{d-1}(\delta/2)$ contains at most $k$ points from $P$ for any $i\in \Integers$.
Furthermore, recall that we have a separator $\sps_i$ between every pair of consecutive points $p_i$ and $p_{i+1}$.
These separators form the set $\SPS \mydef \{\sps_1,\ldots,\sps_{n-1}\}$.
Since any two separators that induce the same partitioning of $P$ are considered equivalent, $\SPS$ is the set of all nontrivial separators.
Recall that we also selected a separator from every second non-empty drum.
Let $\spt_i$ denote the separator in the $i$'th such drum.
We define a function $\ind$ that maps the index $i$ of a separator in $\SPT$ to the index of the corresponding separator in $\SPS$. Thus $\sps_{\ind(i)} = \spt_i$.
Let $\SPT=\{t_1,\ldots,t_{|\SPT|}\}$ be the collection of selected separators.
For every separator $\spt_i$, we computed a family $\B_i$ of candidate endpoint sets in $2^{O(k^{1-1/d})} \cdot n$ time.

We will now place separators in the other nonempty drums.
Each drum contains at most $k$ points, and has width $\delta$.
Therefore, for all $i$ we can find a separator $\spu_i$ in the nonempty drum between $\spt_i$ and $\spt_{i+1}$ such that $x_{\indp(i)+1}-x_{\indp(i)} > \delta/(2k)$, where $\indp$ is defined by $\sps_{\indp(i)} = \spu_i$ for $1 \leq i \leq |\SPT|-1$.
See Figure~\ref{fig:sparsev2:blocks} for an example.
Note that this is not tight;
we can in fact obtain $x_{\indp(i)+1}-x_{\indp(i)} \geq \delta/(k+1)$, but the above will suffice.
Let $\SPU \mydef \{\spu_1, \ldots, \spu_{|\SPT|-1}\}$ be the resulting set of separators.
\begin{figure}
\begin{center}
\includegraphics[]{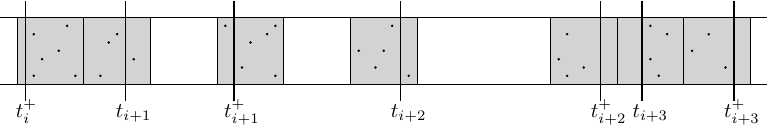}
\caption{The separators $\spt_i$ and $\spu_i$ are placed alternately. Note that even though $\spu_{i+2}$ is placed in a different drum than $\spt_{i+3}$, in this example they are equivalent, since they induce the same partitioning of $P$.}
\label{fig:sparsev2:blocks}
\end{center}
\end{figure}
Note that no two separators $\spu_i, \spu_j$ are equal.
However, we can have $\spt_i = \spu_i$ or $\spu_i = \spt_{i+1}$, namely when $\spt_i$ (or $\spt_{i+1}$) and $\spt_i^+$ induce the same partitioning as the boundary between the consecutive drums in which they are placed.
See Figure~\ref{fig:sparsev2:blocks} for an example.
In this case, we will still consider $\spu_i$ to be to the right of $\spt_i$ and to the left of $\spt_{i+1}$.
Recall that we also defined $\spt_0$ and $\spt_{|\SPT|+1}$ as `separators' which are to the left and to the right of all points, respectively.
To ensure we have no two $\spt_i, \spt_{i+1}$ without a $\spu_{i}$ in between, we define two new separators.
We define $\spu_0 \mydef \spt_0$, but we do consider it to be to the right of $\spt_0$.
Similarly, $\spu_{|\SPT|} \mydef \spt_{|\SPT|+1}$,  which we consider to be to the left of $\spt_{|\SPT|+1}$.
The following observation, which is a slightly modified version of Observation~\ref{obs:random:long_edge_then_bitonic}, shows how we can use the separators from $\SPU$.
\begin{observation}\label{obs:sparsev2:long_edge_then_bitonic}
Let $k$ be such that $P$ contains at most $k$ points in any drum.
Let $\spu_i \in \SPU$.
Let $r_1,r_2$ be two points with $x(r_1) + 2k \delta < x_{\indp(i)}$ and $x_{\indp(i)+1} + 2k \delta < x(r_2)$. 
Let $T$ be an optimal tour on $P$, and suppose that $T$ contains $r_1 r_2$.
Then $T$ is bitonic at $\spu_i$.
\end{observation}
The proof is analogous to the proof of Observation~\ref{obs:random:long_edge_then_bitonic};
see Appendix~\ref{sec:app_omitted-proofs} for the full proof.
Finally, we will need the following observation:
\begin{observation}\label{obs:sparsev2:sparsebound}
Let $k$ be such that $P$ contains at most $k$ points in any drum.
Let $p_i$ and $p_j$ be two points with $i < j$.
Suppose $j - i \geq (m+1)k$ for some $m \in \mathbb{N}$.
Then $x_j - x_i \geq m \delta$.
\end{observation}
\begin{proof}
Since for all $i \in \mathbb{Z}$ the drum~$\sigma_i$ contains at most $k$ points, $x_j - x_i \geq \floor{ \frac{j-i}{k}-1 } \delta$.
From this, the observation directly follows.
\end{proof}

\mypara{Step 2: The three tables.}
Let $\spt_i$ be a separator in $\SPT$.
Recall that the basic algorithm used a family $\B_i$ of sets, such that $\B_i$ contains the endpoint configuration of the edges of an optimal tour that are crossing $t_i$.
The set $\B_i$ was formed as follows (see also Corollary~\ref{cor:candidate-enumeration}).
Let $\R_{1a}$ be the family of sets of points formed by taking the right endpoints of edges of the candidate edge sets from Lemma~\ref{lem:random_det_sep}.
Let $\R_{1b}$ be the family of sets of at most $c^*$ points between $\spt_i$ and $\spt_{i+3}$, where $c^*$ is as defined in Lemma~\ref{lem:sep_cross}.
To improve readability, we write $\R_1 \mydef \{R_{1a} \cup R_{1b} | R_{1a} \in \R_{1a} \text{ and } R_{1b} \in \R_{1b}\}$.
Finally, let $\R_2$ be the family of sets containing at most one point to the right of $\spt_{i+3}$ and none to its left, so each element of $\R_2$ is either the empty set or a set containing a single point.
Then
\[\B_i = \{R_1 \cup R_2 \: | \: R_1 \in \R_1 \text{ and } R_2 \in \R_2\}.\]
See Figure~\ref{fig:sparsev2:tables}(i) for an example. \medskip

Recall that the basic algorithm takes $2^{O(\delta^{1-1/d})} n^2$ time because for each of the $O(n)$ separators $\spt_i$, the family $\B_i$ contains $2^{O(k^{1-1/d})} n$ sets.
To counter this, we replace $\B_i$ by two other families of sets of points, each containing only $2^{O(k^{1-1/d})}$ sets.
Recall that $P[i,j] \mydef \{p_i,\ldots,p_j\}$ for all $i \leq j$. 

Let $\R_2'$ be the family of sets of at most one point from $P[\ind(i)+1, \ind(i)+20k^2]$.
Note that $\R_2'$ contains $20k^2+1$ sets: $20k^2$ singleton sets and the empty set.
The number $20$ is not special;
it is merely a large enough constant.
Finally, let $\R_3$ be the family of sets of at most one point from $P[\ind(i)-20k^2+1,\ind(i)]$.

Now we define $\B_i^{(1)}$ as the family of sets of points formed by combining one set from $\R_1$, one from $\R_2'$ and one from $\R_3$, such that there exists a $j$ such that all the points lie in $P[j,j+20k^2]$.
In other words,
\[\B_i^{(1)} \mydef \{B = R_1 \cup R_2' \cup R_3 \: | \: R_1 \in \R_1, R_2' \in \R_2', R_3 \in \R_3 \text{ and } B \subseteq P[j,j+20k^2] \text{ for some } j\}.\]
See Figure~\ref{fig:sparsev2:tables}(ii) for an example.
Note that the empty set is an element of $\B_i^{(1)}$.
Also, note that not every element of $\B_i^{(1)}$ is an element of $\B_i$.
Specifically, this is true for the elements which contain a point to the left of $\spt_i$. \medskip

%----------------------------------------------------

Let $\R_2''$ be the family of sets of exactly one point from $P[\ind(i) + 20k^2+1,\ind(i) + 20k^2+3k]$.
We define
\[\B_i^{(2)} \mydef \{R_1 \cup R_2'' \: | \: R_1 \in \R_1 \text{ and } R_2'' \in \R_2''\}.\]
Note that $\B_i^{(2)} \subseteq \B_i$ and that $\B_i^{(1)} \cap \B_i^{(i)} = \emptyset$.
See Figure~\ref{fig:sparsev2:tables}(iii) for an example.

Now we can define the tables $A^{(1)}$, $A^{(2)}$ and $A^{(3)}$ that our algorithm uses.
Recall that for Algorithm~\ref{alg:sparse}, we used
\[A[i,B]:= \begin{cases}\parbox{0.8\textwidth}{A representative set containing pairs $(M,x)$, where $M$ is a perfect matching on~$B \in \B_i$
and $x$ is a real number equal to the minimum total length of a path cover of
$P_0\cup\dots\cup P_{i-1} \cup B$ realizing the matching $M$.}\end{cases}\]
We use the same definition for $A^{(1)}[i,B]$ for all $B \in \B_i^{(1)}$, and for $A^{(2)}[i,B]$ for all $B \in \B_i^{(2)}$.
We get
\[A^{(1)}[i,B]:= \begin{cases}\parbox{0.8\textwidth}{A representative set containing pairs $(M,x)$, where $M$ is a perfect matching on~$B \in \B_i^{(1)}$
and $x$ is a real number equal to the minimum total length of a path cover of
$P_0\cup\dots\cup P_{i-1} \cup B$ realizing the matching $M$}\end{cases}\]
and
\[A^{(2)}[i,B]:= \begin{cases}\parbox{0.8\textwidth}{A representative set containing pairs $(M,x)$, where $M$ is a perfect matching on~$B \in \B_i^{(2)}$
and $x$ is a real number equal to the minimum total length of a path cover of
$P_0\cup\dots\cup P_{i-1} \cup B$ realizing the matching $M$.}\end{cases}\]
Note that the only difference between these three definitions is for which $B$ the table entries are defined.
Furthermore, note that $A^{(1)}[||\SPT|+1,\emptyset]$ is defined, and the corresponding representative set contains a single element $(\emptyset, x)$ where $x$ is the length of the shortest tour on all points of $P$.
\medskip

Unfortunately, we cannot compute every value of $A^{(1)}$ and $A^{(2)}$ using only these two tables.
This brings us to the third and final table.
Let $\spu_i \in \SPU$, and let $q_1$ be a point between $\spu_i$ and $\spu_{i+1}$.
Let $q_2$ be a point in $P[\ind(i)+20k^2+1,n]$.
We then define
\[A^{(3)}[i,q_1,q_2]:= \begin{cases}\parbox{0.8\textwidth}{The length of the shortest path from $q_1$ to $q_2$ that visits all points in $P[1, \ind(i)]$, such that the neighbour of $q_2$ is a point in $P[1,\ind(i)-5k^2]$.}\end{cases}\]
See Figure~\ref{fig:sparsev2:tables}(iv) for an example.
\begin{figure}[p]
\begin{center}
\includegraphics[]{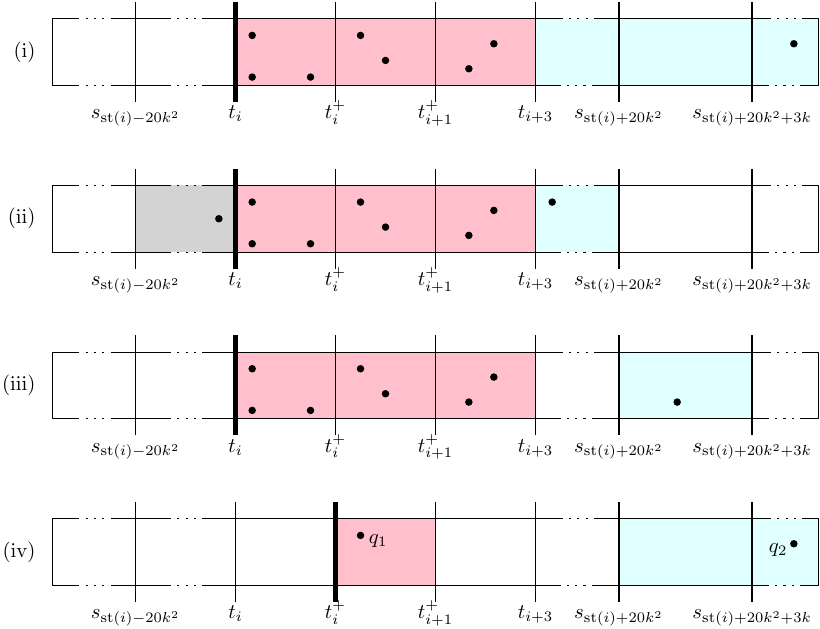}
\caption{Examples of entries of every table.
Note that all three different types of separators are used ($\sps, \spt$ and $\spu$).
Also note that the spaces between consecutive separators contain varying numbers of points.
\protect\\
(i) For comparison, an entry of table $A$ of the basic algorithm.
A set of points in $\B_i$ consists of a set from $\R_1$ (in red), and possibly any point to the right of $\spt_{i+3}$ (in blue).
\protect\\
(ii) An entry of $A^{(1)}$.
A set of points in $\B_i^{(1)}$ contains points from a set in $\R_1$ (in red), possibly a point between $\spt_{i+3}$ and $\sps_{\ind(i)+15k^2}$ (in blue), and possibly a point between $\sps_{\ind(i)-15k^2}$ and $\spt_j$ (in grey).
Not shown is the property that there exists a $j$ such that all the points lie in $P[j,j+20k^2]$.
\protect\\
(iii) An entry of $A^{(2)}$.
A set of points in $\B_i^{(2)}$ contains points from a set in $\R_1$ (in red), and exactly one point between $\sps_{\ind(i)+20k^2}$ and $\sps_{\ind(i)+20k^2+3k}$ (in blue).
\protect\\
(iv) An entry of $A^{(3)}$.
It contains exactly one point $q_1$ between $\spu_i$ and $\spu_{i+1}$ (in red), and exactly one point $q_2$ to the right of $\sps_{\ind(i)+20k^2}$ (in blue).
Not shown is the property that the neighbour of $q_2$ is a point in $P[1,\ind(i)-5k^2]$.
}
\label{fig:sparsev2:tables}
\end{center}
\end{figure}

\mypara{Step 3: Computing the table values.}
Since the algorithm uses three different tables, the order in which to compute these is nontrivial.
Furthermore, for the third table, we will need to precompute two (significantly smaller) tables.
Algorithm~\ref{alg:sparsev2} gives our algorithm in pseudocode, which we will explain below.

\begin{algorithm}
\caption{NarrowRectTSP-DP-v2($P, \delta$)}\label{alg:sparsev2}
\hspace*{\algorithmicindent} \textbf{Input:} A sparse point set $P$ in $(-\infty,\infty) \times \ball^{d-1}$\\
\hspace*{\algorithmicindent} \textbf{Output:} The length of the shortest tour visiting all points in $P$
\begin{algorithmic}[1]
\State Compute the separators $\spt_1,\ldots,\spt_{|\SPT|}$ using Lemma~\ref{lem:random_det_sep}, as explained earlier. 
\State Compute the separators $\spu_1,\ldots,\spu_{|\SPT|-1}$ as explained in Step 1.
\State Compute the sets $\cB_1^{(1)} \! ,\ldots, \cB_{|\SPT|+1}^{(1)}$ as explained in Step 2.
\State Compute the sets $\cB_1^{(2)} \! ,\ldots, \cB_{|\SPT|+1}^{(2)}$ as explained in Step 2.
\State $A^{(1)}[0,\emptyset] \mydef \{(\emptyset,0)\}$

\For {$i = 1$ to $|\SPT|+1$}

	\ForAll {$B\in \B_i^{(1)}$}
		\State Compute $A^{(1)}[i,B]$ as explained in Step 3-a. \label{line:sparsev2:A1}
	\EndFor
	
	\ForAll {$B\in \B_i^{(2)}$}
		\State Compute $A^{(2)}[i,B]$ as explained in Step 3-b. \label{line:sparsev2:A2}
	\EndFor
	
	\If {$1 < i < |\SPT|$}
		\ForAll {$q_1$ between $\spu_i$ and $\spu_{i+1}$}
				
			\ForAll {$q'$ between $\spu_{i-1}$ and $\spu_i$}
				\State Compute $D[i,q_1,q']$, as explained in Step 3-c.\label{line:sparsev2:A3_1}
			\EndFor
		
			\ForAll {$p \in P[\indp(i)-5k^2-5k+1,\indp(i)-5k^2]$}
				\State Compute $E[i,p,q_1]$, as explained in Step 3-c. \label{line:sparsev2:A3_2}
			\EndFor
		
			\ForAll {$q_2 \in P[\ind(i)+20k^2+1,n]$}
				\State Compute $A^{(3)}[i,q_1,q_2]$ as explained in Step 3-c. \label{line:sparsev2:A3_3}
			\EndFor
	
		\EndFor
	\EndIf
\EndFor
\State \Return $\length(A^{(1)}[|\SPT|+1,\emptyset])$
\end{algorithmic}
\end{algorithm}

\mypara{Step 3-a: Computing table values of $A^{(1)}$.}
Computing the values of $A^{(1)}$ is analogous to computing the table values of the basic algorithm.
See Figure~\ref{fig:sparsev2:a1example} for an example.
\begin{figure}
\begin{center}
\includegraphics[]{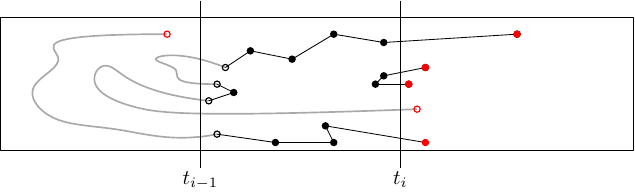}
\caption{An example of the computation of one of the elements of a table entry of $A^{(1)}$.
The red circles (open and filled) form $B$.
The open circles (black and red) form $B'$.
The grey paths contain only points to the left of $t_{i-1}$.}
\label{fig:sparsev2:a1example}
\end{center}
\end{figure}
However, there is one exception.
Sometimes, this would require a (non-existing) table entry $A^{(1)}[i-1,B']$ for some $B' \notin \B^{(1)}_{i-1}$.
Note that this can only be the case because $B'$ contains a point too far to the right.
Since there are at most $3k$ points between $\spt_i$ and $\spt_{i-1}$, the set $B'$ contains a point in $P[\ind(i-1)+20k^2+1,\ind(i-1)+20k^2+3k]$.
Since there exists a $j$ such that $B' \subseteq P[j,j+20k^2]$, the set $B'$ does not contain a point to the left of $\spt_{i-1}$.
Therefore, $B' \in \B^{(2)}_{i-1}$, and so $A^{(2)}[i-1,B']$ is the value we need.

For pseudocode, for a given $i$ and $B \in \B_i^{(1)}$, see Algorithm~\ref{alg:sparsev2-1}.
\begin{algorithm}
\caption{NarrowRectTSP-DP-v2-1$(i,B)$}\label{alg:sparsev2-1}
\begin{algorithmic}[1]
\State $A^{(1)}[i,B] \mydef \emptyset$
\ForAll {$j \in \{1,2\}$}
	\ForAll {$B'\in \cB_{i-1}^{(j)}$ where $B' \subseteq P_{i-1} \cup B$}
		\ForAll {$(M,x)\in \text{\emph{TSP-repr}}(P_{i-1}\cup B' \cup B,\ B'\symdiff B))$}
			\ForAll {$(M',x')\in A^{(j)}[i-1,B']$}		
				\If{$M'$ and $M$ are compatible}
					\State Insert $(\mathrm{Join}(M,M'),x+x')$ into $A^{(1)}[i,B]$
				\EndIf
			\EndFor
		\EndFor
	\EndFor
\EndFor		
\State $Reduce(A^{(1)}[i,B])$
\end{algorithmic}
\end{algorithm}
The analysis of the running time for table $A^{(1)}$ is analogous to the analysis of the basic algorithm.
Therefore, the time needed to calculate a single table entry is $2^{O(k^{1-1/d})}$.

\mypara{Step 3-b: Computing table values of $A^{(2)}$.}
Suppose we want to compute a table entry $A^{(2)}[i,B]$.
Let $p_j$ be the rightmost point of $B$.
Note that $p_j \in P[\ind(i)+20k^2+1,\ind(i)+20k^2+3k]$.
Now, by Observation~\ref{obs:sparsev2:sparsebound}, since
\[j - (\indp(i-1)+1) \geq \ind(i)+20k^2+1 - (\ind(i)+1) = 20k^2 \geq (2k+1)k,\]
we get that $x_j - x_{\indp(i-1)+1} > 2k \delta$.
Therefore, if $p_j$ has a neighbour far enough to the left, then Observation~\ref{obs:sparsev2:long_edge_then_bitonic} gives us that the tour is bitonic at $\spu_{i-1}$.

Let $M$ be a perfect matching on $B$ with corresponding path cover $C$.
Let $p$ be the neighbour of $p_j$ in $C$.
Now there are two possible cases: either $p \in P[\ind(i)-10k^2+1,\ind(i)]$, or $p \in P[1,\ind(i)-10k^2]$.
To compute the value of $x$, we will simply take the minimum of these two cases, illustrated in Figure~\ref{fig:sparsev2:a2example}.
\begin{figure}
\begin{center}
\includegraphics[]{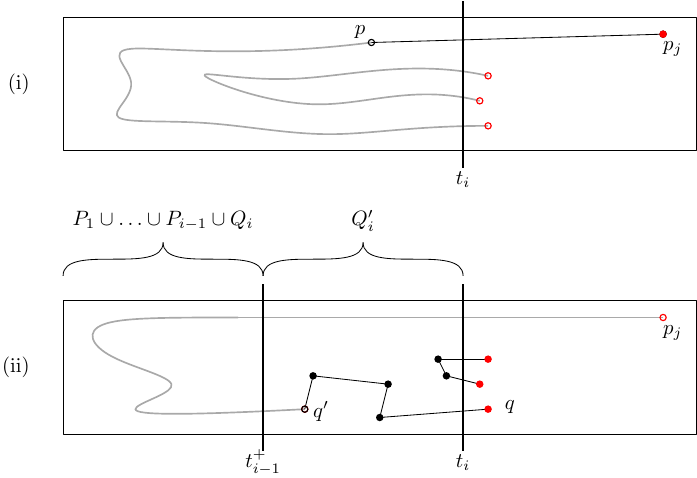}
\caption{The two possible cases for a table entry of $A^{(2)}$.
The red circles (open and filled) form $B$.
The open circles (black and red) form $B'$.
The grey paths contain only points to the left of the leftmost separator shown.}
\label{fig:sparsev2:a2example}
\end{center}
\end{figure}\medskip

\begin{itemize}
\item \textbf{Case (i): $p\in P[\ind(i)-10k^2+1,\ind(i)]$.} For this case, we do the following for all possible $p$, and take the minimum of the resulting values.
We simply add the length of $|p_j p|$ to the corresponding value of $A^{(1)}[i,B']$, where $B'$ is formed by taking $B$, and swapping $p_j$ for $p$.
See Figure~\ref{fig:sparsev2:a2example}(i) for an example.
We claim that $B' \in \B^{(1)}_i$.
Since $p_j$, the rightmost point of $B$, is not in $B'$, the only non-trivial requirement we must check for this claim is that there exists an index $z$ such that $B' \subseteq P[z,z+20k^2]$.
We will now show that this is satisfied for $z = \ind(i)-10k^2$.

We first take a look at the leftmost point of $B'$, which is $p$.
Clearly, $p \in P[\ind(i)-10k^2,\ind(i)+10k^2]$.
Now, the rightmost point from $B'$ is to the left of $\spt_{i+3}$.
Since there are at most $7k$ points between $\spt_i$ and $\spt_{i+3}$, all points except $p$ are guaranteed to be in $P[\ind(i)+1, \ind(i) + 7k+1] \subseteq P[\ind(i),\ind(i)+10k^2]$.
Since both the leftmost point and the rightmost point of $B'$ are in $P[\ind(i)-10k^2, \ind(i) + 10k^2]$, we conclude that $B' \in \B^{(1)}_i$. 

\item \textbf{Case (ii): $p\in P[1,\ind(i)-10k^2]$.} For this case, things are slightly more complicated.
By Observation~\ref{obs:sparsev2:long_edge_then_bitonic}, we get that $C$ must be bitonic at $\spu_{i-1}$.
Let $M^*$ be a perfect matching on $B$.
Let $q$ be the point to which $p_j$ is matched in $M^*$.
See Figure~\ref{fig:sparsev2:a2example}(ii) for an example.
Note that the path from $q$ to $p_j$ of $C$ crosses $\spu_{i-1}$ at least twice.
Therefore, this path crosses $\spu_{i-1}$ exactly twice, and no other path crosses $\spu_{i-1}$.
Let $q'$ be the last point on this path before $\spu_{i-1}$ is crossed the first time.
Note that $q' = q$ can hold.
We can now split $C$ into two parts of which the lengths can be computed independently: the path from $q'$ to $p_j$, and the rest.
Let us take a look at these two separately:
\begin{itemize}
\item Recall that $P_i$ is the set of points between $\spt_i$ and $\spt_{i+1}$.
We want to know the length of the shortest path from $q'$ to $p_j$ using all points in $P_1 \cup \ldots \cup P_{i-2} \cup Q_{i-1}$, where $Q_{i-1}$ is the set of points between $\spt_{i-1}$ and $\spu_{i-1}$.
This is the value $A^{(3)}[i-1,q',p_j]$.
Note that $p_j$ is far enough to the right to make this a valid table entry.
Indeed, since $p \in P[1,\ind(i)-10k^2]$, and there are at most $3k$ points between $\spu_{i-1}$ and $\spt_i$, we have $p \in P[1,\indp(i-1)-5k^2]$, as required. 
Finally, we claim that $q'$ is indeed between $\spu_{i-1}$ and $\spu_i$.
We know that the tour must be bitonic at $\spu_{i-1}$ and $\spu_i$, and that there is at least one point between these two separators.
If $q'$ is not between $\spu_{i-1}$ and $\spu_i$, we would have two edges crossing both of these separators: the edge following $q'$ in the path from $q'$ to $p$ and the edge from $p$ to $p_j$.
However, since there is at least one point between these two separators, there is also at least one point with at least one neighbour to the left of $\spu_{i-1}$ or to the right of $\spu_i$.
Both lead to a contradiction, as we now have at least three edges crossing a separator where the tour must be bitonic.
We conclude that $q'$ is indeed between $\spu_{i-1}$ and $\spu_i$.
\item Let $Q_{i-1}'$ be the set of points between $t_{i-1}^+$ and $t_i$.
In other words, $Q_{i-1}' \mydef P_{i-1} \setminus Q_{i-1}$.
Let $M$ be the matching $M^*$, with one difference:
instead of $q$ being matched to $p_j$, it is matched to $q'$.
Now we want to know the length of the shortest path cover on $Q_{i-1}' \cup B \setminus \{p_j\}$ realizing the matching $M$.
We can find this value using the method of the standard $2^{O(n^{1-1/d})}$ algorithm.
\end{itemize}
We can compute the total length of $C$ for all possible $q'$, and take the minimum of these values, to find the minimal length of $C$.
This ends the second case.
\end{itemize}

In conclusion, for a given $i$ and $B \in \B_i^{(2)}$, we can compute $A^{(2)}[i,B]$ using the pseudocode of Algorithm~\ref{alg:sparsev2-2}.
Notice that the $Reduce$ at the end automatically takes the minimum of all options added at Line~\ref{line:sparsev2:A2:opt1} for every $M$.
\begin{algorithm}
\caption{NarrowRectTSP-DP-v2-2$(i,B)$}\label{alg:sparsev2-2}
\begin{algorithmic}[1]
\State $A^{(2)}[i,B] \mydef \emptyset$
\ForAll {$p \in P[\ind(i)-10k^2+1,\ind(i)]$}
	\State $B' \mydef B \text{ with } p_j \text{ swapped for }p$
	\ForAll {$(M',x') \in A^{(1)}[i,B']$}
		\State $M \mydef M' \text{ with } p \text{ swapped for }p_j$
		\State Insert $(M,x' + |p_j p|)$ into $A^{(2)}[i,B]$. \label{line:sparsev2:A2:opt1}
	\EndFor
\EndFor

\ForAll {$q' \in Q_{i-1}' \cup B \setminus \{p_j\}$}
	\State {$B' \mydef \{q',p_j\}$}
	\State $M' \mydef \{\{q',p_j\}\}$
	\State $x' \mydef A^{(3)}[i-1,q',p_j]$
	\ForAll {$(M,x)\in \text{\emph{TSP-repr}}(Q_{i-1}' \cup B' \cup B,\ B'\symdiff B))$}
		\If{$M'$ and $M$ are compatible}
			\State Insert $(\mathrm{Join}(M,M'),x+x')$ into $A^{(2)}[i,B]$
		\EndIf
	\EndFor
\EndFor
\State $Reduce(A^{(2)}[i,B])$
\end{algorithmic}
\end{algorithm}

Next, let us take a look at the time needed to compute a single table entry of $A^{(2)}$.
For the first part, there are $O(k^2)$ possible $p$.
For every $p$, we need $|A^{(1)}[i,B']| = 2^{O(k^{1-1/d})}$ time.
For the second part, there are $O(k)$ possible $q'$.
For each of these combinations, we run the $2^{O(n^{1-1/d})}$ algorithm on a set containing $O(k)$ points.
Analogously to the basic algorithm, the \textit{Reduce} algorithm runs in $2^{O(k^{1-1/d})}$ time.
Therefore, the total running time needed to compute a single table entry is $k^2 2^{O(k^{1-1/d})} + k 2^{O(k^{1-1/d})} + 2^{O(k^{1-1/d})} = 2^{O(k^{1-1/d})}$.

\mypara{Step 3-c: Computing table values of $A^{(3)}$.}
Finally, we will show how to compute the value of $A^{(3)}[i,q_1,q_2]$.
The point $q_2$ is far enough to the right of $\spu_{i-1}$ to make Observation~\ref{obs:sparsev2:long_edge_then_bitonic} relevant.
Therefore, we will once more split into two possible cases, of which we will use the minimum of the resulting values.
Let $p$ be the neighbour of $q_2$ in the shortest path corresponding to this table entry.
Now either $p \in P[1,\indp(i-1)-5k^2]$, or $p \in P[\indp(i-1)-5k^2+1,\ind(i)-5k^2]$. \medskip

\begin{itemize}
\item \textbf{Case (i): $p \in P[1,\indp(i-1)-5k^2]$.}
From Observation~\ref{obs:sparsev2:sparsebound} we get that $x_{\indp(i-1)} - x(p) > 2k \delta$.
Therefore, by Observation~\ref{obs:sparsev2:long_edge_then_bitonic}, the path from $q_1$ to $q_2$ must be bitonic at $\spu_{i-1}$.
Since the edge $pq_2$ crosses $\spu_{i-1}$, it can only be crossed once before.
See Figure~\ref{fig:sparsev2:a3example}(i) for an example.
Let $q'$ be the last point on this path before $\spu_{i-1}$ is crossed the first time.
Note that $q' \neq q_1$, as $q'$ is between $\spu_{i-1}$ and $\spu_i$, since there is at least one point between these separators.
\begin{figure}
\begin{center}
\includegraphics[]{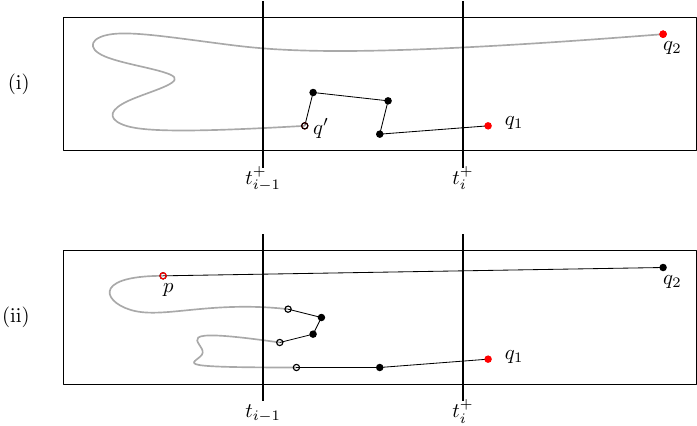}
\caption{The two options for a table entry of $A^{(3)}$.
The red circles (open and filled) form $B$.
The open circles (black and red) form $B'$.
The grey paths contain only points to the left of $\spt_{i-1}$.}
\label{fig:sparsev2:a3example}
\end{center}
\end{figure}

We can now split the value of $A^{(3)}[i,q_1,q_2]$ into the sum of two values which can be computed independently: the length of the path from $q_1$ to $q'$, and the length of the path from $q'$ to $q_2$.
Let us take a look at these two separately:
\begin{itemize}
\item We want to know $D[i,q_1,q'] \mydef $ the length of the shortest path from $q_1$ to $q'$ using all points between $\spu_{i-1}$ and $\spu_i$.
We can find this value using the standard $2^{O(n^{1-1/d})}$ algorithm.
Note that this value is independent of $q_2$, and therefore only needs to be computed once.
This is done at Line~\ref{line:sparsev2:A3_1} of Algorithm~\ref{alg:sparsev2}.
\item We want to know the length of the shortest path from $q'$ to $q_2$ using all points in $P_1 \cup \ldots \cup P_{i-1} \cup Q_{i-1} \cup \{q',q_2\}$.
This value can be found in the table entry $A^{(3)}[i-1,q',q_2]$.
Notice that this is indeed a valid table entry, since $q'$ is between $\spu_{i-1}$ and $\spu_i$.
\end{itemize}
We can compute the sum of these lengths for all possible $q'$, and take the minimum of these sums to find the optimal path length for this case.

\item \textbf{Case (ii): $p \in P[\indp(i-1)-5k^2+1,\ind(i)-5k^2]$.}
In this case, $A^{(3)}[i,q_1,q_2]$ is equal to $|q_2 p|$ plus the length of the shortest path from $p$ to $q_1$, using every point in $P[1,\ind(i)]$.
See Figure~\ref{fig:sparsev2:a3example}(ii) for an example.
Let $B = \{p,q_1\}$.
We claim that table $A^{(1)}$ contains a representative set containing pairs $(M,x)$, where $M$ is a perfect matching on $B$ and $x$ is a real number equal to the total length of the path cover of $P[1,\ind(i)] \cup B$ realizing the matching $M$.
Note that since $B \subseteq P[\indp(i-1)-10k^2+1,\indp(i-1)+3k]$, no invalid table entries of $A^{(1)}$ are called.
Since $B$ contains only two points, this representative set contains only one pair, of which $M = \{\{p,q_1\}\}$.
Therefore, the value $x$ is the length of the shortest path we are looking for.
Since this value is independent of $q_2$, it needs only to be computed once.
This can be done with the following pseudocode, which is the full version of Line~\ref{line:sparsev2:A3_2} of Algorithm~\ref{alg:sparsev2}:
\begin{algorithm}
\begin{algorithmic}[1]
\State $E[i,p,q_1] \mydef \infty$
\State $B \mydef \{p,q_1\}$
\ForAll {$B'\in \B_i^{(1)}$ where $B' \subseteq Q_i \cup B$}
	\ForAll {$(M,x)\in \text{\emph{TSP-repr}}(Q_i \cup B' \cup B,\ B'\symdiff B))$}
		\ForAll {$(M',x')\in A^{(1)}[i,B']$}			
			\If{$M'$ and $M$ are compatible}
				\State $E[i,p,q_1] \mydef \min\{E[i,p,q_1], x+x'\}$
			\EndIf
		\EndFor
	\EndFor
\EndFor
\end{algorithmic}
\end{algorithm}
This ends the second case.
\end{itemize}

Combining the two cases, for a given $i, q_1$ and $q_2$ we get the following pseudocode to use at Line~\ref{line:sparsev2:A3_3} of Algorithm~\ref{alg:sparsev2}:
\begin{algorithm}
\begin{algorithmic}[1]
\State $A^{(3)}[i,q_1,q_2] \mydef \infty$
\ForAll {$q'$ between $\spu_{i-1}$ and $\spu_i$}
	\State $A^{(3)}[i,q_1,q_2] \mydef \min\{A^{(3)}[i,q_1,q_2], \: D[i,q_1,q'] + A^{(3)}[i-1,q',q_2] \}$
\EndFor

\ForAll {$p \in P[\indp(i-1)-5k^2+1,\ind(i)]$}
	\State $A^{(3)}[i,q_1,q_2] \mydef \min\{A^{(3)}[i,q_1,q_2], \: E[i,p,q_1] + |p q_2| \}$
\EndFor
\end{algorithmic}
\end{algorithm}

Finally, let us take a look at the time needed to compute a single table entry of $A^{(3)}$.
For the first option, there are $O(k)$ possible $q'$, and the computations now only take constant time per $q'$.
For the second option, there are $O(k)$ possible $p$ (since $\ind(i) - \indp(i-1) \leq 3k$), and the computations only take constant time per $p$.
This brings the total time needed to compute a single table entry of $A^{(3)}$ to $O(k + k) = O(k)$.
This does not include the computations done at Line~\ref{line:sparsev2:A3_2} and Line~\ref{line:sparsev2:A3_3} of Algorithm~\ref{alg:sparsev2}; these will be considered below.

This concludes the computation of the table entries.

\mypara{Step 4: The running time.}
Let us analyse the total time needed by Algorithm~\ref{alg:sparsev2}.
\begin{itemize}
\item \textbf{Computation of table entries of $A^{(1)}$ (Line~\ref{line:sparsev2:A1}).}
There are $O(n)$ possible $i$ and $2^{O(k^{1-1/d})}$ possible $B \in \B_i^{(1)}$.
As seen before, the computation of a single table entry takes $2^{O(k^{1-1/d})}$ time.
This brings the total time needed for this line to $n 2^{O(k^{1-1/d})} 2^{O(k^{1-1/d})} = 2^{O(k^{1-1/d})} n$.
\item \textbf{Computation of table entries of $A^{(2)}$ (Line~\ref{line:sparsev2:A2}).}
There are $O(n)$ possible $i$ and $2^{O(k^{1-1/d})}$ possible $B \in \B_i^{(2)}$.
As seen before, the computation of a single table entry takes $2^{O(k^{1-1/d})}$ time.
This brings the total time needed for this line to $n 2^{O(k^{1-1/d})} 2^{O(k^{1-1/d})} = 2^{O(k^{1-1/d})} n$.
\item \textbf{Precomputation of table $D$ (Line~\ref{line:sparsev2:A3_1}).}
We compute the length of the shortest path from $q_1$ to $q'$ using all points between $\spu_{i-1}$ and $\spu_i$.
There are $O(n)$ possible $i$, there are $O(k)$ possible $q_1$ and there are $O(k)$ possible $q'$.
The $2^{O(n^{1-1/d})}$ algorithm is run on a point set of size $O(k)$.
Therefore, the total time needed for this line is $2^{O(k^{1-1/d})} n$.
\item \textbf{Precomputation of table $E$ (Line~\ref{line:sparsev2:A3_2}).}
Let $i, p, q_1$ be some valid combination for the second option.
There are $O(n)$ possible $i$, there are $O(k)$ possible $p$ and $O(k)$ possible $q_1$.
Analogously to the computation of a table entry of $A^{(1)}$, for each combination $2^{O(k^{1-1/d})}$ time is needed.
Therefore, this line takes $2^{O(k^{1-1/d})} n$ time in total.
\item \textbf{Computation of table entries of $A^{(3)}$ (Line~\ref{line:sparsev2:A3_3}).}
There are $O(n)$ possible $i$, a total of $O(k)$ possible $q_1$ and $O(n)$ possible $q_2$.
As seen before, the computation of a single table entry takes $O(k)$ time.
This brings the total time needed for this line to $O(n k n k) = O(k^2 n^2)$.
\end{itemize}
This brings the total running time to $2^{O(k^{1-1/d})} n + 2^{O(k^{1-1/d})} n + 2^{O(k^{1-1/d})} n + 2^{O(k^{1-1/d})} n + O(k^2 n^2) = 2^{O(k^{1-1/d})} n + O(k^2 n^2)$.
Since $k = O(\delta)$, this concludes the proof of Theorem \ref{thm:sparsev2:runtime}.
%------------------------------------------------------------------------------------------

%------------------------------------------------------------------------------------------
\section{Concluding remarks} \label{sec:conc}
%------------------------------------------------------------------------------------------
Our paper contains three main results on \etsp.

First, we proved that for points with distinct integer $x$-coordinates
in a strip of width~$\delta$, an optimal bitonic tour is optimal overall when $\delta\leq 2\sqrt{2}$.
The proof of this bound, which is tight in the worst case, is partially automated to reduce the potentially enormous number of cases to two worst-case scenarios.
It would be interesting to see whether a direct proof can be given for this fundamental result.
We note that the proof of Theorem \ref{thm:bitonic:main} can easily be adapted to point sets of which the $x$-coordinates of the points need not be integer, as long as the difference between $x$-coordinates of any two consecutive points is at least 1.

Second, we gave a $2^{O(\delta^{1-1/d})} n + O(\delta^2 n^2)$ dynamic programming algorithm for sparse point sets in a $d$-dimensional $\delta$-cylinder.
For $\delta=\Theta(n)$ the running time becomes $2^{O(n^{1-1/d})}$, which is optimal under ETH.
We expect that the factor $O(\delta^2 n^2)$ can be improved to a factor $O(\delta^2 \, n \log^2 n)$, using a method similar to that used by De Berg \etal\ to solve {\sc Bitonic TSP} in $O(n \log^2 n)$ time \cite{DBLP:journals/corr/BergBJW16}.

Third, we gave a $2^{O(\delta^{1-1/d})}n$ expected time algorithm for random point sets.
The proof also gives a way to relate the expected running times of algorithms for any problem on two different kinds of random point sets:
a version where the $x$-coordinates of the points are taken uniformly at random from $[0,n]$, and a version where the differences between two consecutive $x$-coordinates are taken independently from $\Exp(1)$.
Note that the problem is scalable, so these results generalize to any desired interval or desired expected difference between two consecutive $x$-coordinates.

\medskip

There are several directions for further research. 
In $\Reals^2$, our results show how the complexity of TSP scales as we go from an essentially 1-dimensional problem ($\delta$ is a small constant) to an essentially unrestricted 2-dimensional problem ($\delta=n$).
We can also ask how the complexity of TSP scales from a 2-dimensional problem to a 3-dimensional problem.
For instance, we can assume the points are randomly distributed inside a box of size~$n\times n \times \delta$, or we can consider sparse point sets inside a slab of width $\delta$.
(Here a point set is sparse when every box of size~$1\times 1\times \delta$ contains $O(1)$ points.)
The goal would be to obtain an algorithm with running time $2^{O(f(\delta)\sqrt{n})}$, where $f(n)= O(n^{1/6})$.
Such a running time becomes $2^{O(\sqrt{n})}$ for constant $\delta$ (which is optimal for TSP in $\Reals^2$, under ETH), and it becomes $2^{O(n^{2/3})}$ for $\delta=n$ (which is optimal for TSP in $\Reals^3$, assuming ETH). 
One can also generalize the problem in another way.
Instead of assuming that the points are all close to a line or line segment, one could also assume that the points are all close to for example a curve, or a (connected) set of line segments.
We believe that our algorithm can serve as the basis of an algorithm solving such a problem, under the assumption that the point sets are dense enough to ensure that the solution will generally follow these curves / segments. Making this precise, and investigating how the running time depends on the number of line segments, would be interesting.

More generally, we believe this more fine-grained look at the dimension of a problem instance is worth investigating for many other problems as well.
We already investigated this for the {\sc Minimum Rectilinear Steiner Tree} problem in $\Reals^2$, where the fastest known algorithm runs in $2^{O(\sqrt{n} \log n)}$ time \cite{rectiSteiner16}.
We showed that there exists an $n^{O(\sqrt{\delta})}$ algorithm for sparse point sets and a $2^{O(\delta \sqrt{\delta})} n$ expected time algorithm for random point sets \cite{MRST}.

\appendix
\newpage
%-----------------------------------------------------------------------------------------
\section{The automated prover} \label{sec:app_prover}
Here, we will give an overview of the inner workings of the automated prover.
See Algorithm~\ref{alg:autoprover} for pseudocode, which we will now explain.
First, all candidate sets of edges $\overline{F}'$ that form a tour together with $\widetilde{E}$ are generated.
Then, a set of cases is generated.
Each case consists of the following:
\begin{itemize}
\item The $x$-coordinates of each of the points, stored in $\mathcal{X}$.
Every point must have a different integer $x$-coordinate, and must adhere to the given bounds on the $x$-coordinates.
\item For every point, an interval denoting the range in which its $y$-coordinate must lie, stored in $\mathcal{Y}$.
When the cases are first generated, this is simply $[0,\delta]$ for every point.
\end{itemize}
For each possible combination $\mathcal{X}$ of $x$-coordinates of the points, one such case is generated.
Note that every possible set of points in represented by one of these cases.
Then, for every case, it does the following:
\begin{itemize}
\item To prove a case, it calculates an upper bound on the total lengths of every candidate $\overline{F}'$, and a lower bound of the total length of $\overline{F}$.
For every edge, a lower bound can be found by simply taking the points as close together as possible, and an upper bound is found by doing the opposite.
See Figure~\ref{fig:autoprover:interval_arithmetic} for some examples.
Suppose there exists a set $\overline{F}'$ such that the upper bound on its length is guaranteed to be lower than the lower bound on the length of $\overline{F}$.
Then $\|\overline{F}'\| < \|\overline{F}\|$ must hold, and the case holds.
To counter rounding errors, the calculated upper bound must be more than some fixed $\eta \mydef 10^{-6}$ lower than the calculated lower bound for this to trigger.
The prover has been implemented in Java using \texttt{doubles}, which have a precision of up to 15 to 16 decimal digits.
\item If it fails to prove the case, it splits the case into $2^{n_\myleft + n_\myright}$ cases, by splitting the interval of every $y$-coordinate into two equal parts.
For example, if a case has two points and intervals $[0,2]$ and $[6,8]$, it is split into the four cases $([0,1],[6,7]), ([0,1],[7,8]), ([1,2],[6,7])$ and $([1,2],[7,8])$.
Then, it recursively tries to prove all of these cases.
If, however, the intervals become too small (smaller than the given precision parameter $\eps$), the prover gives up on proving this case.
\end{itemize}
This process continues until all cases have been either proven or given up on.
Finally, the prover returns a list of all cases and whether it succeeded or failed on these cases.
Note that the smaller the precision parameter $\eps$ is, the more precise the cases (and therefore the answer) will be.
However, a smaller $\eps$ will also result in more cases, which increases both the running time and the number of lines of the output.

\begin{algorithm}
\caption{\emph{FindShorterTour}$(n_\myleft,n_\myright,\overline{F},\widetilde{E},X,\delta,\eps)$}\label{alg:autoprover}
\vspace*{2mm}
\textbf{Input:} 
\vspace*{-8mm}
\begin{quotation}
\noindent
\begin{itemize}
\item $n_\myleft$ and $n_\myright$ denote the number of points with $x$-coordinate at most $-1$ and at least $0$, respectively;
\item $\overline{F},\widetilde{E}$ are edge sets such that $\widetilde{E} \cup \overline{F}$ forms a tour;
\item $X$ is an array where $X[i]$ specifies the set from which the $x$-coordinate of the $i$-th point in the given scenario may be chosen; 
\item $\eps$ denotes the maximum `size' of the returned cases.
\end{itemize} 
\end{quotation}
\textbf{Output:} 
\vspace*{-2mm}
\begin{quotation}
\noindent A list of \emph{scenarios} and an \emph{outcome} for each scenario. 

A scenario contains for each point~$q$ an $x$-coordinate~$x(q)$ from the set of allowed $x$-coordinates for~$q$, and a range~$\yrange(q)\subseteq [0,2\sqrt{2}]$ for its $y$-coordinate.
This $y$-range is an interval of length at least~$\eps/2$.
The outcome of a scenario is either \success or \fail.
An outcome \success means that a set $\overline{F}'$ has been found with the desired properties:
$\widetilde{E}\cup \overline{F}'$ is a tour, and for all possible instantiations of the scenario---that is, all choices of $y$-coordinates from the $y$-ranges in the scenario---we have $\|\overline{F}'\| < \|\overline{F}\|$.
An outcome \fail means that such an $\overline{F'}$ has not been found.
It does not guarantee that such an $F'$ does not exist for this scenario.

The list of scenarios is complete in the sense that for any instantiation of the input case there is a scenario that covers it.
\end{quotation}
\begin{algorithmic}[1]
\State $result \leftarrow \emptyset$
\State Generate all $\overline{F}'$ for which $\widetilde{E} \cup \overline{F}'$ is a tour
\For {Every valid combination $\mathcal{X}$ of $x$-coordinates for the points}
	\State Generate case $(\mathcal{X},\mathcal{Y})$, consisting of the exact $x$-coordinates of the points, and the bounds $[0, \delta]$ for the $y$-coordinate of each of the points
	\State \textsc{TryToProveCase($\mathcal{X},\mathcal{Y}$)}
\EndFor

\State \Return $result$
\end{algorithmic}
\hspace*{\algorithmicindent}
\begin{algorithmic}[1]
\Function{TryToProveCase}{$\mathcal{X},\mathcal{Y}$}
	\If{there exists an $\overline{F}'$ such that $\overline{F}'$ is guaranteed to be shorter than $\overline{F}$}
		\State Add $(\mathcal{X},\mathcal{Y},$ {\sc Success}$)$ to $result$
	\ElsIf{the bounds in $\mathcal{Y}$ are larger than $\eps$}
		\For{every of the $2^{n_\myleft + n_\myright}$ subcases $\mathcal{Y'}$ of $\mathcal{Y}$}
			\State \textsc{TryToProveCase($\mathcal{X},\mathcal{Y'}$)}
		\EndFor
	\Else
		\State Add $(\mathcal{X},\mathcal{Y},$ {\sc Fail}$)$ to $result$
	\EndIf
\EndFunction
\end{algorithmic}
\end{algorithm}

\begin{figure}
\begin{center}
\includegraphics{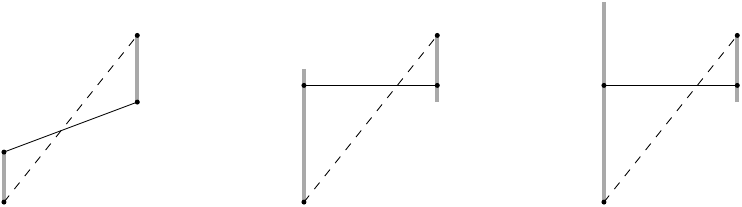}
\caption{Three pairs of intervals, the shortest edges between points in those intervals (solid) and the longest edges between points in those intervals (dashed)}
\label{fig:autoprover:interval_arithmetic}
\end{center}
\end{figure}

%------------------------------------------------------------------------------------------

\newpage

%------------------------------------------------------------------------------------------
\section{Omitted proofs}\label{sec:app_omitted-proofs}

%-------------------------------------------------------------------------
\subsection{Proof of Observation~\ref{obs:sparsev2:long_edge_then_bitonic}}
%-------------------------------------------------------------------------
See Figure~\ref{fig:algproofs:long_edge_then_bitonic} for an example.
Note that $x_{\indp(i)}$ is the $x$-coordinate of the rightmost point to the left of $\spu_i$, and that $x_{\indp(i)+1}$ is the $x$-coordinate of the leftmost point to the right of $\spu_i$.
Therefore, $x_{\indp(i)+1} - x_{\indp(i)} > \delta / (2k)$.
\begin{figure}
\begin{center}
\includegraphics{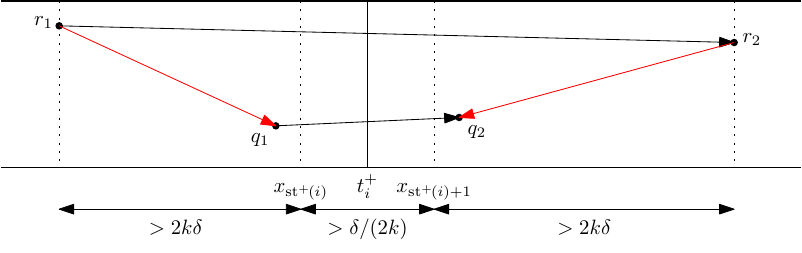}
\caption{An example of Observation~\ref{obs:sparsev2:long_edge_then_bitonic}.
Edges of $T$ are shown in black, edges of $T'$ are shown in red.}
\label{fig:algproofs:long_edge_then_bitonic}
\end{center}
\end{figure}

W.l.o.g., $T$ contains the directed edge from $r_1$ to $r_2$.
Suppose for a contradiction that $T$ is an optimal tour and not bitonic at $\spu_i$.
Then $T$ must cross $\spu_i$ at least four times, which implies that there must be another directed edge $q_1 q_2$ with $x(q_1) \leq x_{\indp(i)}$ and $x_{\indp(i)+1} \leq x(q_2)$.

Note that $q_1 \neq r_1$ and $q_2 \neq r_2$, since a single point cannot have two incoming or two outgoing edges.
We will now show that $| r_1 q_1| + |r_2 q_2| < |q_1 q_2| + |r_1 r_2|$.
By Observation~\ref{obs:2ktonicproofs:switcheroo}, we then get an optimal tour $T'$ shorter than $T$, thereby finishing the proof.
First, we will simplify the problem using an argument similar to the proof of Observation~\ref{obs:2ktonicproofs:linesum}.
Suppose we move $q_1$ towards $q_2$ until $x(q_1) = x_{\indp(i)}$.
By doing so, $|q_1 r_1|$ will never decrease more than $|q_1 q_2|$.
Therefore, it is sufficient to prove the statement for $x(q_1) = x_{\indp(i)}$.
Analogously, it is sufficient to prove the statement for $x(q_2) = x_{\indp(i)+1}$ and $x(r_1) + 2k \delta = x(q_1)$ and $x(q_2) + 2k \delta = x(r_2)$.
Finally, recall that $x_{\indp(i)+1} - x_{\indp(i)} > \delta / (2k)$, so we will prove the statement for $x_{\indp(i)+1} - x_{\indp(i)} = \delta / (2k)$.
Together, we now have
$$x(r_1) + 2k\delta = x(q_1) = x(q_2)-\delta/(2k) = x(r_2)-\delta/(2k)-2k\delta.$$
We get
\begin{align*}
|q_1 r_1| + |q_2 r_2|
    & \leq \sqrt{(x(q_1)-x(r_1))^2+\delta^2} + \sqrt{(x(r_2)-x(q_2))^2+\delta^2} \\
    & = 2 \delta \sqrt{ (2k)^2 + 1} \\
    & < 2 \delta (2k + 1/(2k))\\
    & = \delta/(2k) + (2 (2k\delta) + \delta/(2k)) \\
    & = (x(q_2)-x(q_1)) + (x(r_2)-x(r_1)) \\
    & \leq |q_1 q_2| + |r_1 r_2|.
\end{align*}
By Observation~\ref{obs:2ktonicproofs:switcheroo}, this gives us an optimal tour $T'$ shorter than $T$, thereby finishing the proof.

\bibliographystyle{plain}
\bibliography{bibfile}

\end{document}